\documentclass[conference]{IEEEtran}
\IEEEoverridecommandlockouts
\usepackage{cite}
\usepackage{amsmath,amssymb,amsfonts,amsthm}
\usepackage{algorithmic}
\usepackage{graphicx}
\usepackage{textcomp}
\usepackage{xcolor}
\def\BibTeX{{\rm B\kern-.05em{\sc i\kern-.025em b}\kern-.08em
    T\kern-.1667em\lower.7ex\hbox{E}\kern-.125emX}}
\usepackage{balance}  
\usepackage[ruled,vlined,linesnumbered]{algorithm2e}
\usepackage{tikz}
\usepackage{enumitem}
\usepackage[labelfont=bf,skip=0pt]{caption,subfig}
\usepackage{hyperref}
\usepackage{xcolor}
\usepackage{euscript}
\usepackage{xparse}
\usepackage{mathrsfs}

\newcommand{\pnt}{p}

\newcommand{\Qinterval}{I}
\newcommand{\interval}[2]{[#1, #2]}
\newcommand{\wtv}{\textbf{u}}

\renewcommand{\Re}{\mathbb{R}}

\newcommand{\cardin}[1]{| #1 | }
\newcommand{\floor}[1]{\left \lfloor #1 \right \rfloor}
\newcommand{\ceiling}[1]{\left \lceil #1 \right \rceil}

\newcommand{\remove}[1]{}
\newcommand{\Expec}[1]{\mathbf{E}\left [ #1 \right ]}
\newcommand{\Prob}[1]{\mathbf{Pr}\left[\,#1\,\right]}

 \newenvironment{proofSk}{\trivlist\item[]\emph{Proof (Sketch)}.}%
                   {\unskip\nobreak\hskip 1em plus 1fil\nobreak%
                           \rule{2mm}{2mm}
                           \parfillskip=0pt%
                           \endtrivlist}

\newtheorem{lemma}{Lemma}

\newtheorem{example}{Example}[section]

\usepackage{euscript}
\newcommand{\narrow}[1]{\protect\scalebox{0.7}[1.0]{\ensuremath{\textsf{#1}}}}

\renewcommand{\emptyset}{\ensuremath{\varnothing}}

\DeclareMathOperator{\polylog}{polylog}

\usepackage{xparse}
\usepackage{xifthen}

\newcommand{\card}[1]{\ensuremath{\lvert#1\rvert}}
\newcommand{\Times}{\ensuremath{\mathbb{T}}}
\newcommand{\kskyband}[2]{
	\ensuremath{\mathfrak{S}^{#1}\big({#2}\big)}
}
\def\mparagraph#1{\par\medskip\noindent\textbf{#1.}\quad}
\newcommand{\topk}[2]{\pi_{#1}(#2)}
\newcommand{\score}{f_\wtv}

\newcommand{\DurTop}{\ensuremath{\narrow{DurTop}}}

\begin{document}

\title{Durable Top-K Instant-Stamped Temporal Records with User-Specified Scoring Functions
\\(Technical Report Version)}

\author{%
\IEEEauthorblockN{Junyang Gao\textsuperscript{\textsection}}
\IEEEauthorblockA{Google Inc. 
}
\and
\IEEEauthorblockN{Stavros Sintos\textsuperscript{\textsection}}
\IEEEauthorblockA{University of Chicago 
}
\and
\IEEEauthorblockN{Pankaj K. Agarwal}
\IEEEauthorblockA{Duke University 
}
\and
\IEEEauthorblockN{Jun Yang}
\IEEEauthorblockA{Duke University
}
}

\maketitle
\begingroup\renewcommand\thefootnote{\textsection}
\footnotetext{Most of the work was conducted when authors were at Duke University.}
\endgroup

\begin{abstract}
A way of finding interesting or exceptional records from instant-stamped temporal data is to consider their ``durability,'' or, intuitively speaking, how well they compare with other records that arrived earlier or later, and how long they retain their supremacy.
For example, people are naturally fascinated by claims with long durability, such as: \emph{``On January 22, 2006, Kobe Bryant dropped 81 points against Toronto Raptors. Since then, this scoring record has yet to be broken.''}
In general, given a sequence of instant-stamped records, suppose that we can rank them by a user-specified scoring function $f$, which may consider multiple attributes of a record to compute a single score for ranking.
This paper studies \emph{durable top-$k$ queries}, which find records whose scores were within top-$k$ among those records within a ``durability window'' of given length, e.g., a 10-year window starting/ending at the timestamp of the record.
The parameter $k$, the length of the durability window, and parameters of the scoring function (which capture user preference) can all be given at the query time.
We illustrate why this problem formulation yields more meaningful answers in some practical situations than other similar types of queries considered previously.  
We propose new algorithms for solving this problem, and provide a comprehensive theoretical analysis on the complexities of the problem itself and of our algorithms. 
Our algorithms vastly outperform various baselines (by up to two orders of magnitude on real and synthetic datasets).
\end{abstract}

\section{Introduction}\label{sec:intro}
Instant-stamped temporal data consists of a sequence of records, each timestamped by a time instant which we call the arrival time, and ordered by the arrival time.
Such data is ubiquitous in a rich variety of domains; i.e., sports statistics, weather measurement, network traffic logs and e-commerce transactions.
A way of finding interesting or unusual records from such data is to consider their ``durability,'' or, intuitively speaking, how well they compare with other records (i.e., records that arrive earlier or later) and how long they retain the supremacy.
For example, consider the performance record: ``On January 22, 2006, Kobe Bryant scored 81 points against Toronto Raptors.''
While impressive by itself, this statement can be boosted by adding some temporal context:
``At that time, this record was the top-1 scoring performance \emph{in the past 45 years of NBA history}.''
Naturally, the further back we can extend the ``durability'' (while the record still remains top), the more convincing the statement becomes.
We can extend durability forward in time as well: ``Since 2006, Kobe's 81 points scoring performance has yet to be broken as of today.''
The notion of durability is widely used in media and marketing, because people are naturally attracted by those events that ``stood the test of time.''
Such analysis of durability is a useful part of the toolbox for anybody who works with historical data,
and can be particularly helpful 
to journalists and marketers in identifying newsworthy facts and communicating their impressiveness to the public.
Because temporal data can accumulate to very large sizes (especially for granular data such as weather or network statistics), and because users often want to find durable records with respect to different ranking criteria quickly, we need to answer durable top-$k$ queries efficiently.

In this paper, we consider \emph{durable top-$k$ queries} for finding instant-stamped records that stand out in comparison to others within a surrounding time window.
In general, each record may have multiple attributes (besides the timestamp) whose values are relevant to ranking these records.
We assume that there is a user-specified scoring function $f$ that takes a record as input, potentially considers its multiple attributes, and computes a single numeric score used for ranking.
Intuitively, a durable top-$k$ query returns, given a time duration $\tau$, records that are within top $k$ during a $\tau$-length time window anchored relative to the arrival time of the record.
How the window should be positioned relative to the arrival time depends on the application; our solution only stipulates that the relative positioning is done consistently across all records.
In practice, we observe most statements in media involving durability either \emph{ends} the window at the arrival time of the record (i.e., \emph{looking back} into the past) or \emph{begins} the window at the arrival time of the record (i.e., \emph{looking ahead} into the future).
Generally speaking, each record returned by our durable top-$k$ corresponds to a statement about the record that highlights the durability of its supremacy.

\begin{figure}[t]
    \centering
    \subfloat[Rebound highlights]{\includegraphics[width=0.25\textwidth]{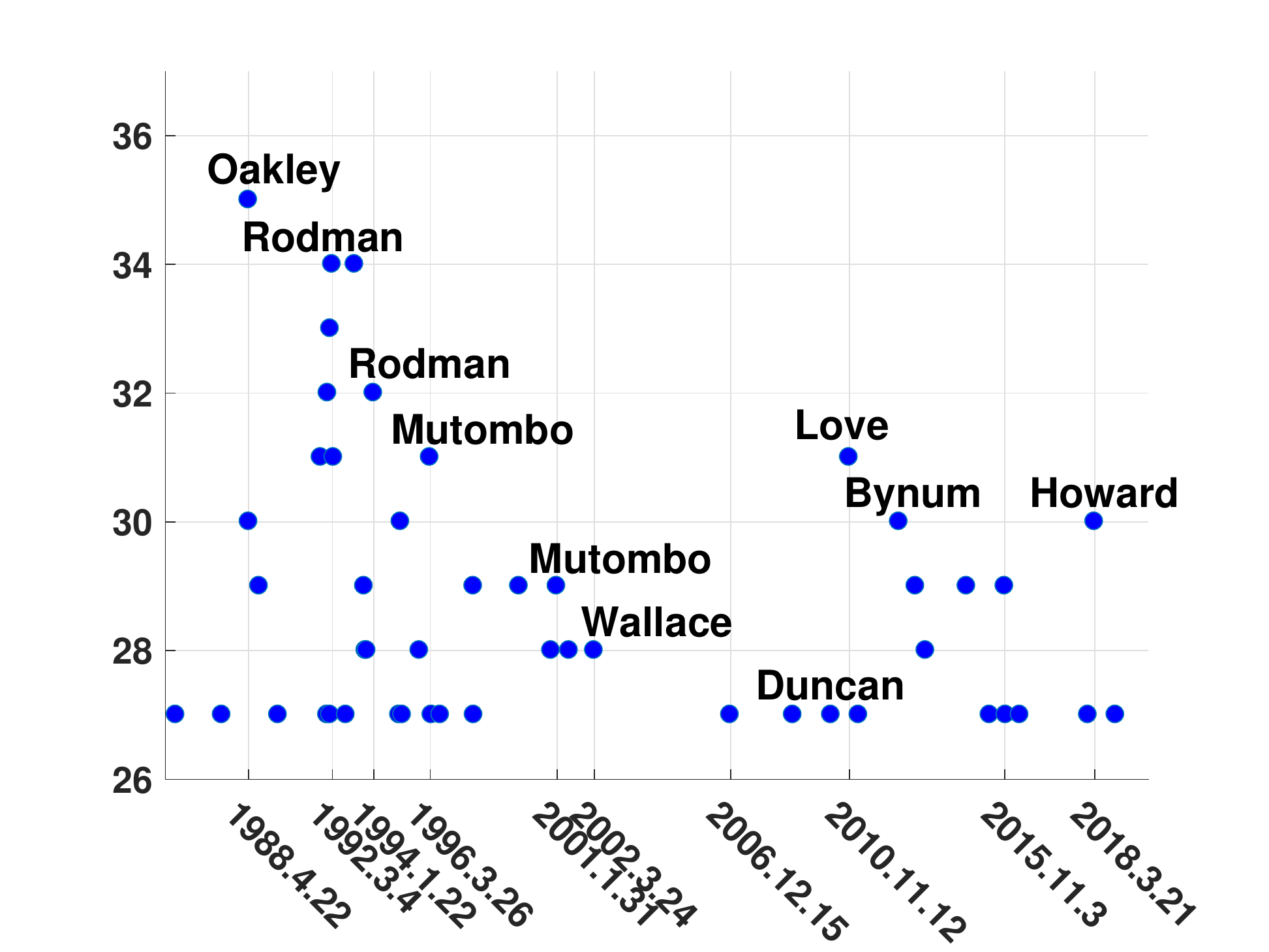}}
    \subfloat[durable top-$k$ query]{\includegraphics[width=0.25\textwidth]{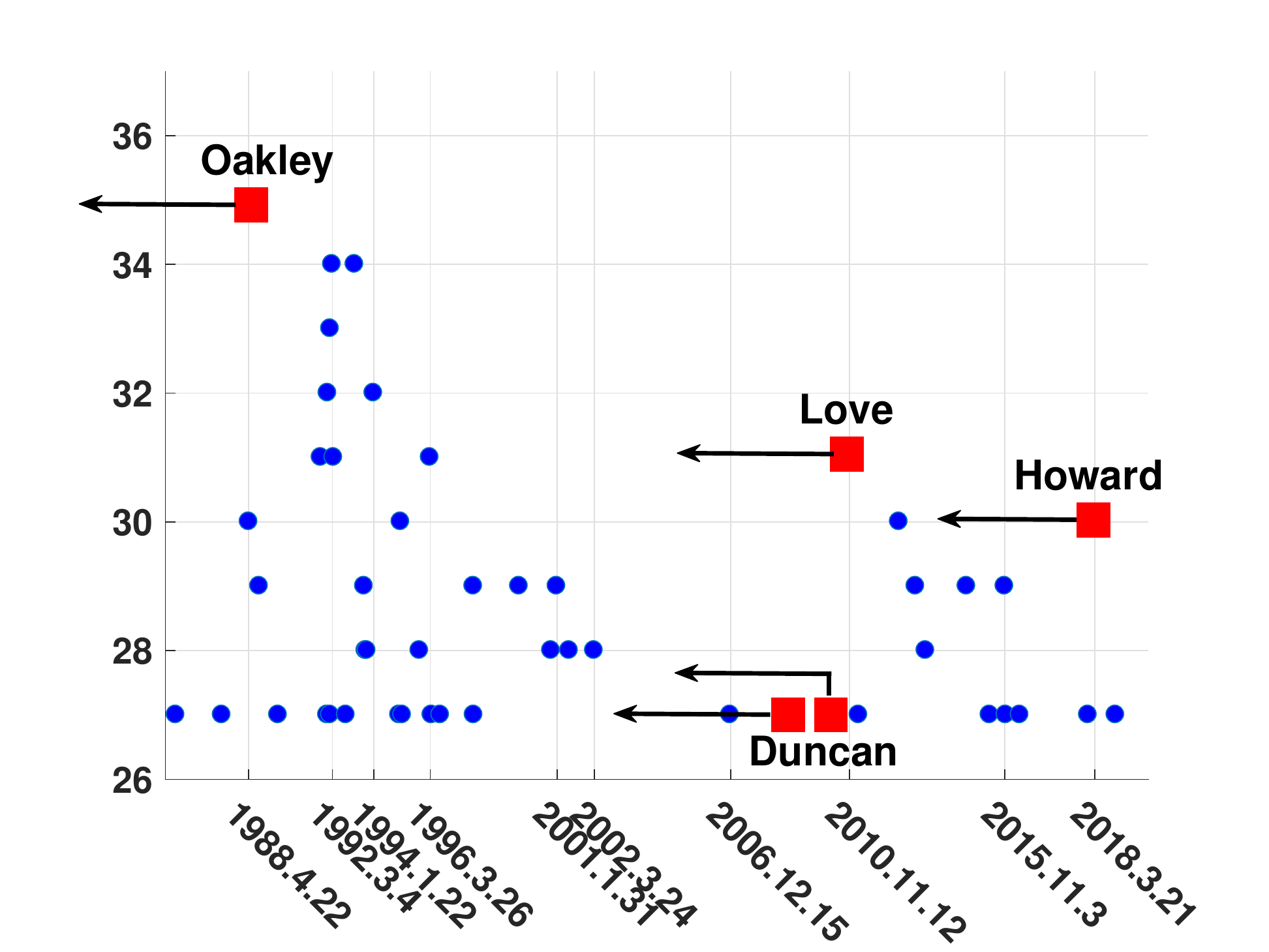}}\\\vspace*{-2.5ex}
    \subfloat[Tumbling Window Top-$k$]{\includegraphics[width=0.25\textwidth]{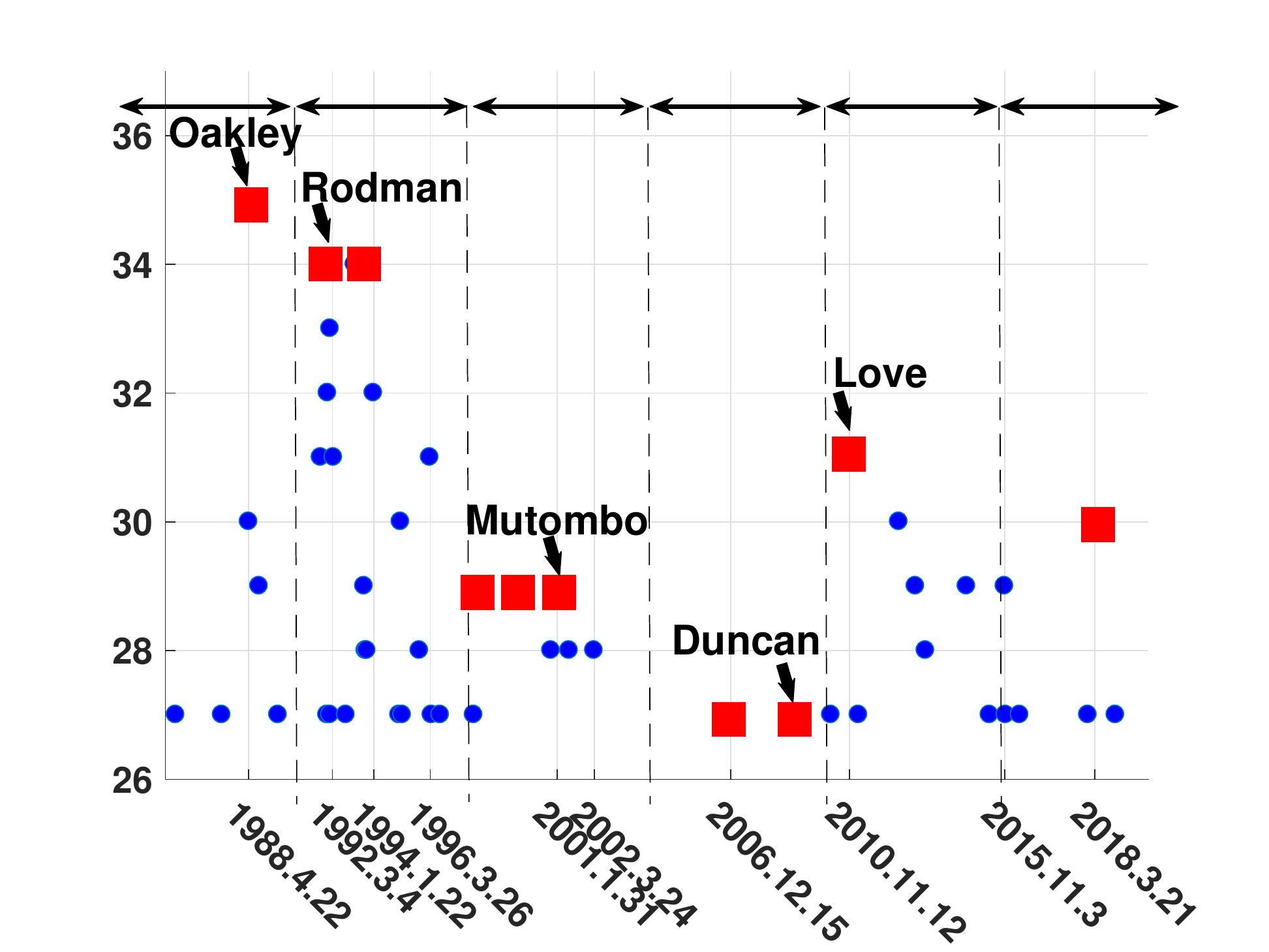}}
    \subfloat[Sliding Window Top-$k$]{\includegraphics[width=0.25\textwidth]{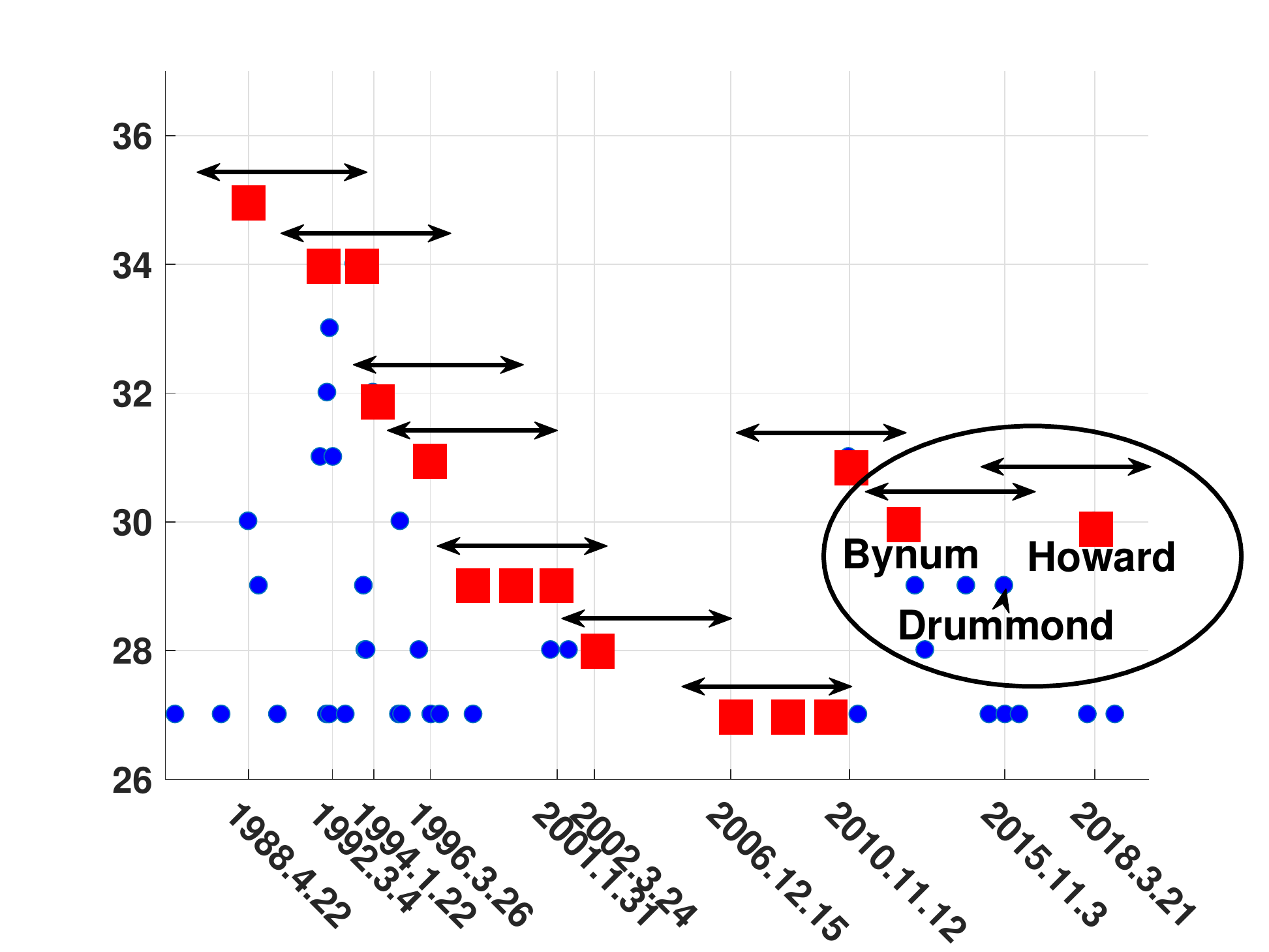}}
    \caption{A case study on finding durable noteworthy rebound performances in NBA history. Red squares highlight results returned by different queries, and line segments represent the durability time window.}
    \label{fig:nba_3pm}
\end{figure}

Note that there are different ways for capturing the notion of durability in queries, including some types that have been studied in the past.
Different application scenarios may call for different semantics.
To understand why our definition of durable top-$k$ queries may be more appropriate than others in some scenarios, we examine the alternatives with a simple concrete example.
\begin{example}\itshape\label{example:1}
Suppose we are interested in finding exceptional rebounds performances (by individual players in individual games) in NBA history---particularly, those that stood out as the top record (or tying for the top record) in a 5-year time span.
Figure~\ref{fig:nba_3pm}.(1) plots all relevant records (i.e., no fewer than 27 rebounds by a single player in a single game) in entire NBA history.
We consider the following three queries to accomplish our task; the latter two have been widely studied in the stream processing and top-$k$ query processing literature.
Note that in this example $k = 1$.
\begin{itemize}[leftmargin=*]
    \item \textbf{Durable top-$k$ (our query)}: This is the query that we propose. 
    For each record, we look back in a 5-year window ending at the timestamp of the record, and check whether the record has the top score among all records within this window.
    Figure~\ref{fig:nba_3pm}.(2) highlights the records (red squares) returned by our query; for each result record, we also show its 5-year durability window as a line segment ending at the record for which it remains on the top.
    \item \textbf{Tumbling-window top-$k$}:
    This query first partitions the timeline into a series of non-overlapping, fixed-sized (5-year) windows, and then returns the top record within each time window.
    The placement of the windows is up to the user and can affect results.
    Results for one particular placement of the windows are shown in Figure~\ref{fig:nba_3pm}.(3).
    \item \textbf{Sliding-window top-$k$}:
    This query slides a 5-year window along the timeline, and returns the top record for each position of the sliding window.
    Figure~\ref{fig:nba_3pm}.(4) highlights a few representative sliding windows, as well as the top records during these windows.
\end{itemize}
All these queries are able to uncover some meaningful durable top records; i.e., for any data record $(X,Y,Z)$ marked as a red square in Figure~\ref{fig:nba_3pm}, we can claim ``player $X$ grabbed $Y$ rebounds in a game on date $Z$, which is the best in \emph{\bfseries some} 5-year span.''
First, the durability aspect adds to the impressiveness of the statement.
Second, the combination of durability and ranking helps reveal interesting records that would otherwise be ignored if we simply filter the records by a high absolute value.
For instance, all three queries find (Duncan, 27, 2009) as a durable top-1 record.
While this record may not seem impressive by number alone, it was indeed the top-1 from 2002 to 2010.
This is an interesting observation, as it reflects a trend (relatively low rebounds of all players) during that era of NBA.

However, there are also notable differences.
\begin{itemize}[leftmargin=*]
\item \textbf{Tumbling-window vs.\ our query}:
The general observation is that the results of tumbling-window are highly sensitive to the choice of window placement.
In Figure~\ref{fig:nba_3pm}.(3), tumbling-window picks (Mutombo, 29, 2001) and the other two performances with 29 rebounds as they were the best ones during 2000-2005, but there were more impressive performances right before them, unfortunately leaving the impression that they stood out only because the windows were cherry-picked.
Furthermore, if we choose to place all windows slightly to the right such that the last window ends with the most recent arrival time, (Rodman, 34, 1992) will be eliminated by (Oakley, 35, 1988), and (Duncan, 27, 2009) will be overlooked since it is shadowed by (Love, 31, 2010).
Overall, because of high sensitivity to window boundaries, tumbling-window runs the risk of omitting important records as they happen to be overshadowed by some other records in the same window, and picking less interesting records as they happen to be the top ones in that specific window.

%
\item \textbf{Sliding-window vs.\ our query}:
Sliding-window is not susceptible to window placement, but it effectively considers all possible window placements, and it returns the union of all top records for each such placement.
This approach leads to possibly many records that are not as meaningful in practice.
In Figure~\ref{fig:nba_3pm}.(4), sliding-window apparently returns overwhelmingly more results compared to our query, which makes it less applicable to mining most noteworthy records.
Even more unnatural is the fact that as we slide the window along the timeline, a record can come in and out of the result; i.e., there is no continuity.
To illustrate, suppose we are interested in durable top-2 records with 5-year windows, and let us focus on Drummond's 29 rebounds performance on 2015.11.3 (highlighted in Figure~\ref{fig:nba_3pm}.(4)).
It is surrounded by two top performance (Howard, 30, 2018) and (Bynum, 30, 2013).
Sliding-window will return this record when the window is positioned at 2014-2019, but \emph{not} when positioned at 2013-2018;
however, the record will be returned again when the window moves to 2012-2017.
Such discontinuity makes the results rather unnatural to interpret.
\end{itemize}
In comparison, our query does not have the issue of sensitivity to window placement or that of difficulty of interpretation, because we assess each record in a 5-year window that leads up to its own timestamp.
Thus, our query result records can be consistently interpreted as having durability ``within the past 5 years'' and clearly communicated to the audience.
The results from the other two queries would be qualified with rather specific durability windows,%
\footnote{
A related question is whether we can post-process the results of the sliding-window query to obtain the results to our query; e.g., filtering those result records in Figure~\ref{fig:nba_3pm}.(4) to get those in Figure~\ref{fig:nba_3pm}.(2).
Unfortunately, such an approach, which we consider as one of the baselines in our experiments, is prohibitively slow on large datasets, as we shall show in later sections.
}
which may be perceived as cherry-picking.
In general, we argue that consistency and simplicity of our query make it more applicable to journalists, marketers, and data enthusiasts alike who seek result that are easily explainable to the public.
\end{example}
In comparison, our query does not have the issue of sensitivity to window placement or that of difficulty of interpretation, because we assess each record in a 5-year window that leads up to its own timestamp.
Thus, our query result records can be consistently interpreted as having durability ``within the past 5 years'' and clearly communicated to the audience.
The results from the other two queries would be qualified with rather specific durability windows,
which may be perceived as cherry-picking.

Although the above example ranks records by a single attribute, its argument can be extended to the general case where records are ranked by a user-specified scoring function that combines multiple attribute values into a single score.

Besides sports, durable top-$k$ queries have applications across many other domains.
For instance, Wikipedia states that ``In late January 2019, an extreme cold wave hit the Midwestern United States, and brought the coldest temperatures in the past 20 years to most locations in the affected region, including some all-time record lows.''
This statement stems from a simple durable top-$k$ query over historical weather data, and allows the Wikipedia article to convey the severity of event effectively.
As an example involving more complex ranking, cybersecurity analysts rely on network traffic log to identify unusual and potentially malicious intrusions.
With a appropriately defined scoring function that combines multiple features of a session, such as duration, volume of data transfer, number of login attempts, and number of servers accessed, a durable top-$k$ query can quickly help identify unusual traffic (relative to others around the same time) for further investigation.
As another example, a financial broker may accompany a recommendation with a statement ``The price-to-earnings ratio (P/E) of this stock last Friday was among the top 5 P/E's within its section for more than 30 days,'' which is also a durable top-$k$ query.
In sum, the efficiency of durable top-$k$ queries makes them suitable for using large volumes of historical efficiently to drive insights or identify leads for further investigation;
the conceptual simplicity of these queries also make them particular attractive for explaining insights and communicating them effectively to the public.

\mparagraph{Contributions}
Our contributions are as follows:
\begin{itemize}[leftmargin=*]
    \setlength{\itemsep}{0em}
    \item We propose to find ``interesting'' records from large instant-stamped temporal datasets using durable top-$k$ queries.
    Compared with other query types related to durability, our query produces results that are more robust (i.e., less sensitive to window placement than tumbling-window) and more meaningful (i.e., easier to interpret than sliding-window).

    \item We propose a suite of solutions based on two approaches that process ``promising'' records in different prioritization orders.
    We provide a comprehensive theoretical analysis on complexities of the problem and of our proposed solutions.
    
    \item Our solutions are general and flexible.
    They do not dictate any specific scoring function $f$, but instead assume a well-defined building block for answering top-$k$ queries using $f$, which can be ``plugged into'' our solutions and analysis.
    We give some concrete example of $f$ and the building block in later sections.
    In particular, $f$ can be further parameterized according to user preference; these parameters, along with $k$, $\tau$ and $\Qinterval$ (the overall temporal range of history of interest), can be specified at query time, making our solutions flexible and suitable for scenarios where users may explore parameter setting at run-time, interactively or automatically.
    
    \item We show that the query time complexity of our algorithms is proportional to $O(\card{S} + k\big\lceil\tfrac{\card{\Qinterval}}{\tau}\big\rceil)$ in the worst case, where $\card{S}$ is the answer size.
    Furthermore, we prove that the expected answer size of a durable top-$k$ query $\card{S}$ is $O(k\big\lceil\tfrac{\card{\Qinterval}}{\tau}\big\rceil)$ under the random permutation model (where the data values can be arbitrarily chosen by an adversary but arrival order is random);
this result implies that the expected query time of our algorithms in practice is linear in the output size. 
    
\end{itemize}

\mparagraph{Paper Overview}
In a nutshell, our proposed algorithms 1) visit promising records in some manner, and 2) check the durability (with respect to a top-$k$ query) for each record we visit.
Techniques for improvement mostly focus on how to efficiently identify candidate records and eventually reduce the total number of durability checks in the second step.
Our proposed algorithms come in two flavors: time-prioritized and score-prioritized, introduced in Section~\ref{sec:time-solution} and Section~\ref{sec:weight-solution}, respectively.
The \textbf{\emph{time-prioritized}} solution traverses and finds candidate records sequentially along the timeline, while the \textbf{\emph{score-prioritized}} solution greedily chooses unvisited candidates with the maximum score (with respect to $f$).
Though in different manners, we show in later sections that these two solutions actually equivalently reduce and bound the size of candidate records (or, the number of durability checks).
More interestingly, in Section~\ref{sec:answer-size}, we further demonstrate that the bound is proportional to the answer size of a durable top-$k$ query, which means our algorithms run faster when the query is more selective, e.g., with smaller $k$ or longer durability $\tau$.
Section~\ref{sec:expr} experimentally evaluates our proposed solutions, including implementations inside a database system.
Section~\ref{sec:related-work} reviews related work and Section~\ref{sec:conclusion} concludes.

\section{Problem Statement and Preliminaries}\label{sec:ps}

\begin{table}
\scriptsize
\caption{Table of notation}
\label{Table:Notation}
\centering
 \begin{tabular}{|c|c|} 
 \hline
 $\Times$ & Time domain\\ \hline
 $p.t$ & Arrival time of $p$\\ \hline
 $f$ & Scoring function\\ \hline
 $k$ & Parameter of Top-$k$ query\\ \hline
 $\pi_{\leq k}([t_1, t_2])$ & Top-$k$ records in time interval $[t_1,t_2]$\\ \hline
 $\Qinterval$ & Query interval\\ \hline
 $\tau$ & Durability duration\\ \hline
 $u$ & Query vector\\ \hline
 $s(n), q(n)$ & Space, query time of top-$k$ index\\ \hline
\end{tabular}
\end{table}

\mparagraph{Problem Statement}
Consider a dataset $P$ with $n$ records, where each record $p \in P$ has $d$ real-valued attributes 
and is represented as a point $(p.x_1, p.x_2, \dots, p.x_d) \in \mathbb{R}^d$.  
For simplicity, we consider a discrete time domain of interest $\Times = \{1,2, \dots, n\}$, and let $p.t \in \Times$ denote the \emph{arrival time} of $p$.
All records in $P$ are organized by increasing order of their arrival time.
Given a non-empty time window $W : [t_1, t_2] \subseteq \Times$,
let $P(W)$ denote the set of records that arrive between $t_1$ and $t_2$; i.e., $P(W) = \{p \in P \mid t_1 \leq p.t \leq t_2\}$. 

Assume a user-specified scoring function maps each record $p$ to a real-valued score, $f : \Re^d \rightarrow \Re$.
Given a time window $W = \interval{t_1}{t_2}$, a \emph{top-$k$ query} $Q(k,W)$ asks for the $k$ records from $P(W)$ with the highest scores with respect to $f$.
Let $\topk{\leq k}{[t_1,t_2]}$ denote the result of $Q(k,W)$; i.e., for $\forall p \in \topk{\leq k}{[t_1,t_2]}$, there are no more than $k-1$ records $q \in P([t_1,t_2])$ with $f(q) > f(p)$.


For simplicity of exposition, we consider durability windows ending at the arrival time of each record (i.e., the ``looking-back'' version), but our solution can be extended to the general case where the windows are anchored consistently relative to the arrival times (including the ``looking-ahead'' version).
We say a record $p$ is $\tau$-durable\footnote{If $\tau$ is obvious from the context, we drop $\tau$ from the definition, i.e., we say that a record is durable.} if $p \in \topk{\leq k}{[p.t - \tau,p.t]}$.
That is, $p$ remains in the top-$k$ for $\tau$ time during $[p.t-\tau,p.t]$.
We are interested in finding records with long durability.
Note that if a record $p$ is $\tau$-durable, then it is also $\tau'$-durable for $\tau'\leq \tau$.
We are interested in finding records with ``long enough'' durability, i.e., durability at least $\tau$.
Given a query interval $\Qinterval$ and a durability threshold $\tau\in [1, \card{\Times}]$, a \emph{durable top-$k$} query, denoted $\DurTop(k,\Qinterval, \tau)$, returns the set of $\tau$-durable records that arrive during $\Qinterval$; 
i.e., $\DurTop(k, \Qinterval, \tau) = \{p \in P(\Qinterval) \mid p \in \topk{\leq k}{[p.t - \tau,p.t]} \}$.
For a record $p\in \DurTop(k, \Qinterval, \tau)$ we can also ask what is the maximum duration that it remains in the top-$k$.
Table~\ref{Table:Notation} summarized our notations.

\mparagraph{Scoring Function and Top-$k$ Query Building Block}\label{sec:pre:query-primitives}
As discussed earlier, our proposed algorithms and complexity analyses are applicable to any user-specified scoring function $f$ as long as there exists a ``building block'' that can answer basic (non-durable) top-$k$ queries under $f$.
This building block can be a ``black box'': the novelty and major contribution of our algorithms come from its ability to reduce and bound the number of invocations of the building block, totally independent of how the building block operates itself.
Of course, the overall algorithm complexity still depends on the efficiency of the building block.
For a function $f$, we consider that an index of size $O(s(n))$ can be constructed in $O(u(n))$ time that answers top-$k$ queries with respect to $f$ in $O(q(n)+k)$ time, where $n$ is the data size and $s(\cdot), u(\cdot), q(\cdot)$ are functions of $n$.

In this paper, we are more interested in top-$k$ queries on a subset of data specified by a time window $W$ given at query time; i,e., computing $Q(k, W)$ that reports the $k$ records in $P(W)$ with the highest scores with respect to $f$.
With a slight care, the top-$k$ query building block can be used to solve this problem by paying a logarithmic factor in index size, query time and construction time.
That is, for a function $f$ we  can construct an index of size $O(s(n)\log n)$ in $O(u(n)\log n)$ time so that for given $k, W$, $Q(k, W)$ can be computed in $O((q(n)+k)\log n)$ time.
If the top-$k$ building block supports updates (insertion/deletion of an item) in $O(\alpha(n))$ time, our range top-$k$ index also supports updates in $O(\alpha(n)\log n)$ time.

Here, we give some concrete examples of $f$ that are widely used in real-life applications, for which efficient top-$k$ query building blocks exist.
Consider the following class of scoring functions parameterized by $\wtv$, which captures user preference:
\begin{itemize}[leftmargin=*]
\setlength\itemsep{0em}
    \item \emph{linear}: $\score(p) =  \sum_{i=1}^{d} \wtv_i \cdot p.x_i$,
    \item \emph{linear combination of monotone scoring functions}: $\score(p) = \sum_{i=1}^{d} \wtv_i \cdot h(p.x_i)$, where $h$ is a monotone function; i.e., $h(\cdot) = \log(\cdot)$,
    \item \emph{cosine}: $\score(p) = \frac{1}{\card{p}\card{\wtv}} \sum_{i=1}^{d} \wtv_i \cdot p.x_i$,
\end{itemize}
where $\wtv$ is a real-valued preference vector and $f_\wtv$ denotes that the scoring function $f$ is parameterized by $\wtv$.
We refer to this class of functions as \emph{preference functions}.
Top-$k$ queries using such class of scoring functions (preferably in the above three forms) have been well studied over the past decades both in computational geometry \cite{afshani2009optimal, chazelle1985power, matousek1992reporting, agarwal2000efficient, agarwal1995dynamic, chan2012three} and databases \cite{chang2000onion, yi2003efficient, hristidis2004algorithms, ilyas2008survey}.
For example, for preference functions above, there is an index with $u(n)=O(n)$, $s(n)=O(n)$, and $q(n)=O(n^{1-1/\floor{d/2}})$, skipping $\polylog (n)$ factors.
Using the results in~\cite{agarwal1995dynamic}, updates can also be supported in $\alpha(n)=O(\polylog (n))$ time.

As mentioned above, users can replace the scoring block with other functions (i.e., non-linear or non-monotone). 
The centerpiece of our algorithm and analysis, which bounds the \emph{number of invocations} of the top-$k$ query building block, remains unchanged.
But in that case, the complexity of the building block will affect the overall complexity bound.
We choose these functions because 1) they are widely used in real-life applications that require ranking and 2) they are both linear and monotone, so preference top-$k$ can be efficiently answered (using the same index).

\mparagraph{Sliding-Windows and Baseline Solution}
\textcolor{black}{
Recall from the discussion in Example~\ref{example:1} (Figures~\ref{fig:nba_3pm}-(2) and \ref{fig:nba_3pm}-(4)) that there is a connection between our problem and the sliding-window version, which has been well studied~\cite{mouratidis2006continuous, jin2008sliding,das2007ad}.
Indeed, one of our baseline solution is adopted from ~\cite{mouratidis2006continuous} with incremental top-$k$ maintenance over sliding windows\footnote{In particular, the idea of Skyband Maintanence Algorithm (SMA) to reduce the number of top-$k$ re-computations from scratches.}.
However, the standard sliding-window technique is more suitable for data streams, where incoming data must be scanned linearly anyway.
Instead, our query analyzes historical data. The linear complexity of sliding windows becomes infeasible especially when dealing with large datasets.
The limitation hence motivates our solutions in later sections.
Experimental results demonstrate our algorithms' significant efficiency gain (up to 2 orders of magnitude) over sliding-window baselines. 
}

\mparagraph{Duration of durable top-$k$ records}
When an algorithm finds a record $p$ in $\DurTop(k, \Qinterval, \tau)$, we can also get the maximum duration (in history) that it remains in the top-$k$. We do it by running a binary search with respect to the arrival times of the records back in history. For each step of the binary search we ask a top-$k$ query to check if $p$ is still in the top-$k$ records. The correctness follows from the observation that if a record is $\tau'$-durable then it is also $\tau$-durable for any $\tau<\tau'$. The binary search has $O(\log n)$ steps and each top-$k$ query takes $O(q(n))$ time. For all records in $|\DurTop(k, \Qinterval, \tau)|$ this procedure takes $O(|\DurTop(k, \Qinterval, \tau)|\cdot q(n)\log n)$ time. Notice that this procedure is independent of the algorithm we use to find the $\tau$-durable records in $I$, so it can be applied in the end of all the algorithms we propose in the next sections (without increasing their total running time).
\section{Time-Prioritized Approach}\label{sec:time-solution}
The time-prioritized approach is straightforward: we visit records in time order and check their durability.
We start with a baseline approach (Section~\ref{sec:time-solution:sw}) and propose an improved version (Section~\ref{sec:time-solution:esw}) using the observation that we can skip many unpromising records in practice.
What is more interesting is how this simple improvement leads to provably substantial reduction in complexity (Section~\ref{sec:time-solution:analysis}).

\subsection{Time-Baseline Algorithm}\label{sec:time-solution:sw}
 
We start with a baseline solution, referred to as \emph{Time-Baseline} or \emph{T-Base}.
T-Base shares the same spirit as the solution proposed in ~\cite{mouratidis2006continuous}, where authors studied the problem on how to continuously monitor top-$k$ queries over the most recent data in a streaming setting.
The main idea is to incrementally maintain the top-$k$ set over continuous sliding windows.
We start with the right endpoint of query interval, and sequentially slide a $\tau$-length window backwards along the timeline.
For each sliding window $[t-\tau, t]$, we need the top-$k$ result to check whether the record (arriving at time $t$) is $\tau$-durable.  
With two adjacent windows $W_1 = [t-\tau, t]$ and $W_2 = [t-\tau-1, t-1]$, top-$k$ results could be updated incrementally, if the expired record (e.g., $P[t]$) is not a top-$k$ on $W_1$.
Otherwise, we need to compute the top-$k$ on window $W_2$ from scratch to guarantee correctness.
The procedure repeats until we visit all records in the query interval $\Qinterval$.

Next, we analyze the query time complexity of T-Base.
There are only two types of records: durable or non-durable.
After visiting each durable record, we need to issue a top-$k$ query.
After visiting each non-durable record, we only need to incrementally update the current top-$k$ set with new incoming record in $O(\log k)$ time.
Assuming a top-$k$ query can be answered in $O\big((q(n) + k)\log n\big)$ time, then T-Base runs in $O\big(\card{S}(q(n) + k)\log n + n\log k)\big)$, where $\card{S}$ is the answer size.
This algorithm takes super-linear time (on the number of records in the query interval). 
Next, we show a solution with sub-linear query time.

\subsection{Time-Hop Algorithm}\label{sec:time-solution:esw}
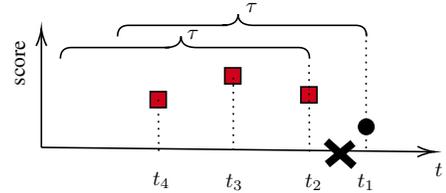
\begin{figure}
    \centering
    \scalebox{0.8}{\tikzset{every picture/.style={line width=0.75pt}} 

\begin{tikzpicture}[x=0.75pt,y=0.75pt,yscale=-1,xscale=1]

\draw    (79.3,96.2) -- (325.3,98.98) ;
\draw [shift={(327.3,99)}, rotate = 180.65] [color={rgb, 255:red, 0; green, 0; blue, 0 }  ][line width=0.75]    (10.93,-3.29) .. controls (6.95,-1.4) and (3.31,-0.3) .. (0,0) .. controls (3.31,0.3) and (6.95,1.4) .. (10.93,3.29)   ;

\draw  [color={rgb, 255:red, 0; green, 0; blue, 0 }  ,draw opacity=1 ][fill={rgb, 255:red, 0; green, 0; blue, 0 }  ,fill opacity=1 ] (279.27,83.21) .. controls (279.28,80.53) and (281.45,78.37) .. (284.12,78.38) .. controls (286.79,78.39) and (288.95,80.56) .. (288.94,83.23) .. controls (288.94,85.9) and (286.76,88.06) .. (284.09,88.06) .. controls (281.42,88.05) and (279.26,85.88) .. (279.27,83.21) -- cycle ;
\draw   (284.3,26.2) .. controls (284.3,21.53) and (281.97,19.2) .. (277.3,19.2) -- (213.3,19.2) .. controls (206.63,19.2) and (203.3,16.87) .. (203.3,12.2) .. controls (203.3,16.87) and (199.97,19.2) .. (193.3,19.2)(196.3,19.2) -- (134.3,19.2) .. controls (129.63,19.2) and (127.3,21.53) .. (127.3,26.2) ;
\draw  [line width=3.75]  (258.93,92.36) -- (276.33,107.99)(275.74,93.52) -- (259.53,106.84) ;
\draw   (248.3,40.2) .. controls (248.3,35.53) and (245.97,33.2) .. (241.3,33.2) -- (177.3,33.2) .. controls (170.63,33.2) and (167.3,30.87) .. (167.3,26.2) .. controls (167.3,30.87) and (163.97,33.2) .. (157.3,33.2)(160.3,33.2) -- (98.3,33.2) .. controls (93.63,33.2) and (91.3,35.53) .. (91.3,40.2) ;
\draw  [fill={rgb, 255:red, 208; green, 2; blue, 27 }  ,fill opacity=1 ] (242.89,58) -- (253,58) -- (253,68.11) -- (242.89,68.11) -- cycle ;
\draw  [fill={rgb, 255:red, 208; green, 2; blue, 27 }  ,fill opacity=1 ] (194.89,46) -- (205,46) -- (205,56.11) -- (194.89,56.11) -- cycle ;
\draw  [fill={rgb, 255:red, 208; green, 2; blue, 27 }  ,fill opacity=1 ] (147.89,61) -- (158,61) -- (158,71.11) -- (147.89,71.11) -- cycle ;
\draw  [dash pattern={on 0.84pt off 2.51pt}]  (284,24.2) -- (284,100.2) ;

\draw  [dash pattern={on 0.84pt off 2.51pt}]  (248,44.2) -- (248.3,99.2) ;

\draw  [dash pattern={on 0.84pt off 2.51pt}]  (152.95,66.05) -- (153.3,99.2) ;

\draw  [dash pattern={on 0.84pt off 2.51pt}]  (200.3,52.2) -- (200.3,98.2) ;

\draw    (79.3,96.2) -- (79.3,23.2) ;
\draw [shift={(79.3,21.2)}, rotate = 450] [color={rgb, 255:red, 0; green, 0; blue, 0 }  ][line width=0.75]    (10.93,-3.29) .. controls (6.95,-1.4) and (3.31,-0.3) .. (0,0) .. controls (3.31,0.3) and (6.95,1.4) .. (10.93,3.29)   ;

\draw (330,110) node  [align=left] {$\displaystyle t$};
\draw (212,8) node  [align=left] {$\displaystyle \tau $};
\draw (284,118) node  [align=left] {$\displaystyle t_{1}$};
\draw (251,118) node  [align=left] {$\displaystyle t_{2}$};
\draw (201,118) node [align=left] {$\displaystyle t_{3}$};
\draw (155,117) node [align=left] {$\displaystyle t_{4}$};
\draw (175,25) node  [align=left] {$\displaystyle \tau $};
\draw (66,43) node [rotate=-270.65] [align=left] {score};

\end{tikzpicture}}
    \caption{Data skipping in Time-Hop Algorithm.}
    \label{fig:asw}
\end{figure}

It is not hard to see that the durable top-$k$ query can be viewed as an \emph{offline} version of the top-$k$ query in the sliding-window streaming model.
Hence, the baseline algorithm introduced above does not best serve our needs.
Since the entire data is available in advance, the manner of continuous sliding window wastes too much time on those non-durable records. 
After all, a meaningful durable top-$k$ query should be selective.

Before describing the algorithm, we illustrate the main idea using an example for $k=3$, shown in Figure~\ref{fig:asw}.
By running a top-3 query $Q(3, [t_1-\tau, t_1])$,
consider the record $p$ arriving at $t_1$ (black circle) is not $\tau$-durable; i.e., $p \not\in \topk{\leq 3}{[t_1-\tau, t_1]}$.
We know the current top-$3$ set contains records (red squares) that arrive at $t_4, t_3$ and $t_2$.
Then, no records arriving between $t_2$ and $t_1$ would be $\tau$-durable and
we can safely hop from $t_1$ to $t_2$.
This simple and useful observation simplifies the query procedure, and allows larger strides for sliding windows.

Now, we present our algorithm \emph{Time-Hop} (T-Hop) (the pseudocode can be found in Algorithm~\ref{algo:asw}).
For each record we visit with timestamp $t_i$, we run a top-$k$ query in $[t_i-\tau, t_i]$ (Line 4).
If the record is not durable, we slide the window back to the \emph{most recent} arrival time of records, say $t_j$, in the current top-$k$ set (Line 9), skipping the non-durable records between $t_j$ and $t_i$.
Otherwise, if a durable record is found, we slide the window backwards by 1 (Line 7) as usual.
Note that if we adopt the look-ahead version of durability, we just need to reverse the traversal order (and time-hopping) on timeline as well.

\begin{algorithm}[t]\small
 \KwIn{$P$, $k$, $\tau$, and $\Qinterval:[t_1,t_2]$.}
 \KwOut{$\DurTop(k, \Qinterval, \tau)$}
 Initialize answer set: $S \leftarrow \emptyset$, top-$k$ set: $\pi_{\leq k} \leftarrow \emptyset$\;
 $t_{curr} \leftarrow t_2$\; 
 
 \While {$t_{curr} >= t_1$} {
   $\pi_{\leq k} \leftarrow Q(k, [t_{curr} - \tau, t_{curr}])$\;
   \uIf {$P[t_{curr}] \in \pi_{\leq k}$} {
   $S \leftarrow S \cup P[t_{curr}]$\;
   $t_{curr} \leftarrow t_{curr}-1$\;
   }
   \Else {
   $t_{curr} \leftarrow$ most recent arrival time of records in $\pi_{\leq k}$\;
   }
 }
 \KwRet S\;
 \caption{T-Hop $(k, \Qinterval, \tau)$\label{algo:asw}}
\end{algorithm}

\subsection{Complexity Analysis of T-Hop}\label{sec:time-solution:analysis}
For the Time-Hop algorithm, the time complexity purely depends on the number of top-$k$ queries called in the query procedure.
We provide a worst-case guarantee on the number of top-$k$ queries performed, as shown by the lemma below (See Appendix~\ref{appendix:proof1} for full proofs).
\begin{lemma}\label{lemma:asw-bound}
The total number of top-$k$ queries performed by the Time-Hop algorithm is $O\big(\card{S} + k\big\lceil\tfrac{\card{\Qinterval}}{\tau}\big\rceil\big)$.
\end{lemma}
\begin{proofSk}
For each record we visit in T-Hop, a top-$k$ query is called for a durability check.
If the record is not $\tau$-durable, we refer it to as a \emph{false check}. Otherwise, we add it to the answer set.
Hence, we only need to bound the total number of false checks.
We decompose the total number of false checks into a set of disjoint $\tau$-length windows, and derive an upper bound of false checks that happen in such a window.

In particular, let $\rho$ be a window of length $\tau$ and let $S_\rho$ be the $\tau$-durable records in $\rho$.
We divide the false checks in $\rho$ into two types. If a false check appears immidiately after a $\tau$-durable record (found by the algorithm) then this is a type-1 false check. Otherwise it is a type-2 false check.
From the definition, the number of type-1 false checks in $\rho$ is $O(S_\rho)$.
Furthermore, we show that after finding $i$ type-2 false checks in $\rho$, a top-$k$ query (that is called for durability check) can only find $k-i$ records in $\rho$. In that way we show that the number of type-2 false checks is $O(k)$.

Given a query interval $\Qinterval$, there are at most $\big\lceil\tfrac{\card{\Qinterval}}{\tau}\big\rceil$ disjoint $\tau$-length sub-intervals.
We conclude that the number of top-$k$ queries is $O\big(\card{S} + k\big\lceil\tfrac{\card{\Qinterval}}{\tau}\big\rceil\big)$.
\end{proofSk}
Overall, with an efficient top-$k$ module, T-Hop answers a durable top-$k$ query $\DurTop(k, \Qinterval, \tau)$ in 
$O\big((\card{S} + k\big\lceil\tfrac{\card{\Qinterval}}{\tau}\big\rceil)(q(n) + k)\log n\big)$
time.
Compared to T-Base, T-Hop runs in sublinear query time (assuming that the ratio $\big\lceil\tfrac{\card{\Qinterval}}{\tau}\big\rceil$ is not arbitrarily large), i.e., the running time does not have a linear dependency on the number of records in $\Qinterval$.
Our experimental results in Section~\ref{sec:expr} suggests that T-Hop is one to two orders of magnitude faster than T-Base in practice. Furthermore, we recall that our index can be implemented with near linear size and polylogarithmic update time for preference queries.

Notice that the number of top-$k$ queries performed by T-Hop depends on $\card{S}$ and $k\big\lceil\tfrac{\card{\Qinterval}}{\tau}\big\rceil$. Ideally, we would like to argue that the number of top-$k$ queries is $O(\card{S})$. In theory, the term $k\big\lceil\tfrac{\card{\Qinterval}}{\tau}\big\rceil$ can be arbitrarily large comparing to $\card{S}$. 
In Section~\ref{sec:answer-size:random} we study the expected size of $S$ in a random permutation model where a set of $n$ scores, chosen by an adversary, are assigned randomly to the records.
In such a case we show that the expected size of $S$ is roughly $O(k\big\lceil\tfrac{\card{\Qinterval}}{\tau}\big\rceil)$, meaning that in practice we expect that the number of top-$k$ queries we execute are asymptotically equal to $\card{S}$.

\vspace{-1em}
\section{Score-Prioritized Approach}\label{sec:weight-solution}
\begin{figure}
    \centering
    \scalebox{0.8}{\tikzset{every picture/.style={line width=0.75pt}} 

\begin{tikzpicture}[x=0.75pt,y=0.75pt,yscale=-1,xscale=1]

\draw    (125.3,97.2) -- (389.3,99.22) ;
\draw [shift={(391.3,99.23)}, rotate = 180.44] [color={rgb, 255:red, 0; green, 0; blue, 0 }  ][line width=0.75]    (10.93,-3.29) .. controls (6.95,-1.4) and (3.31,-0.3) .. (0,0) .. controls (3.31,0.3) and (6.95,1.4) .. (10.93,3.29)   ;

\draw  [color={rgb, 255:red, 0; green, 0; blue, 0 }  ,draw opacity=1 ][fill={rgb, 255:red, 0; green, 0; blue, 0 }  ,fill opacity=1 ] (157.27,29.21) .. controls (157.28,26.53) and (159.45,24.37) .. (162.12,24.38) .. controls (164.79,24.39) and (166.95,26.56) .. (166.94,29.23) .. controls (166.94,31.9) and (164.76,34.06) .. (162.09,34.06) .. controls (159.42,34.05) and (157.26,31.88) .. (157.27,29.21) -- cycle ;
\draw  [color={rgb, 255:red, 0; green, 0; blue, 0 }  ,draw opacity=1 ][fill={rgb, 255:red, 0; green, 0; blue, 0 }  ,fill opacity=1 ] (235.27,39.21) .. controls (235.28,36.53) and (237.45,34.37) .. (240.12,34.38) .. controls (242.79,34.39) and (244.95,36.56) .. (244.94,39.23) .. controls (244.94,41.9) and (242.76,44.06) .. (240.09,44.06) .. controls (237.42,44.05) and (235.26,41.88) .. (235.27,39.21) -- cycle ;
\draw  [color={rgb, 255:red, 0; green, 0; blue, 0 }  ,draw opacity=1 ][fill={rgb, 255:red, 0; green, 0; blue, 0 }  ,fill opacity=1 ] (198.27,60.21) .. controls (198.28,57.53) and (200.45,55.37) .. (203.12,55.38) .. controls (205.79,55.39) and (207.95,57.56) .. (207.94,60.23) .. controls (207.94,62.9) and (205.76,65.06) .. (203.09,65.06) .. controls (200.42,65.05) and (198.26,62.88) .. (198.27,60.21) -- cycle ;
\draw   (295.3,22.2) .. controls (295.31,17.53) and (292.98,15.2) .. (288.31,15.19) -- (242.72,15.12) .. controls (236.05,15.11) and (232.72,12.77) .. (232.73,8.1) .. controls (232.72,12.77) and (229.39,15.1) .. (222.72,15.09)(225.72,15.09) -- (169.01,15) .. controls (164.34,14.99) and (162.01,17.32) .. (162,21.99) ;
\draw   (373.3,35.2) .. controls (373.31,30.53) and (370.98,28.2) .. (366.31,28.19) -- (320.72,28.12) .. controls (314.05,28.11) and (310.72,25.77) .. (310.73,21.1) .. controls (310.72,25.77) and (307.39,28.1) .. (300.72,28.09)(303.72,28.09) -- (247.01,28) .. controls (242.34,27.99) and (240.01,30.32) .. (240,34.99) ;
\draw   (338.3,55.2) .. controls (338.31,50.53) and (335.98,48.2) .. (331.31,48.19) -- (285.72,48.12) .. controls (279.05,48.11) and (275.72,45.77) .. (275.73,41.1) .. controls (275.72,45.77) and (272.39,48.1) .. (265.72,48.09)(268.72,48.09) -- (212.01,48) .. controls (207.34,47.99) and (205.01,50.32) .. (205,54.99) ;
\draw  [dash pattern={on 0.84pt off 2.51pt}]  (162.3,6.2) -- (162.3,97.2) ;

\draw  [dash pattern={on 0.84pt off 2.51pt}]  (202.3,6.2) -- (202.3,97.2) ;

\draw  [dash pattern={on 0.84pt off 2.51pt}]  (240.3,5.2) -- (240.3,101.2) ;

\draw  [dash pattern={on 0.84pt off 2.51pt}]  (294.61,6.2) -- (294.61,98.11) ;

\draw  [dash pattern={on 0.84pt off 2.51pt}]  (339.3,5.2) -- (339.3,97.2) ;

\draw  [dash pattern={on 0.84pt off 2.51pt}]  (373.3,5.2) -- (373.3,99.2) ;

\draw [color={rgb, 255:red, 208; green, 2; blue, 27 }  ,draw opacity=1 ][line width=3]    (242.94,98.23) -- (291.61,98.11) ;

\draw    (126.3,98.2) -- (126.3,25.2) ;
\draw [shift={(126.3,23.2)}, rotate = 450] [color={rgb, 255:red, 0; green, 0; blue, 0 }  ][line width=0.75]    (10.93,-3.29) .. controls (6.95,-1.4) and (3.31,-0.3) .. (0,0) .. controls (3.31,0.3) and (6.95,1.4) .. (10.93,3.29)   ;

\draw (394,84) node  [align=left] {$\displaystyle t$};
\draw (152,16) node  [align=left] {$\displaystyle p_{1}$};
\draw (227,31) node  [align=left] {$\displaystyle p_{2}$};
\draw (188,58) node  [align=left] {$\displaystyle p_{3}$};
\draw (226,5) node  [align=left] {$\displaystyle \tau $};
\draw (317,18) node  [align=left] {$\displaystyle \tau $};
\draw (282,38) node  [align=left] {$\displaystyle \tau $};
\draw (180,116.2) node  [align=left] {1};
\draw (222,116.2) node  [align=left] {2};
\draw (268,116.2) node  [align=left] {3};
\draw (319,116.2) node  [align=left] {2};
\draw (358,116.2) node  [align=left] {1};
\draw (113,45) node [rotate=-270.65] [align=left] {score};

\end{tikzpicture}}
    \caption{Blocking mechanism in score-prioritized approach}
    \label{fig:blocking}
    \vspace{-1em}
\end{figure}
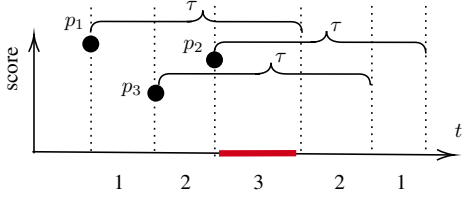
One weakness of time-prioritized approach is that it does not pay much attention to scores and simply visit records sequentially along the timeline (with hops).
Though Lemma~\ref{lemma:asw-bound} shows that T-Hop visits $O(\card{S} + k\big\lceil\tfrac{\card{\Qinterval}}{\tau}\big\rceil)$ records in the worst case, it still potentially visits many low-score and non-durable records and ask more top-$k$ queries.
In contrast, the score-prioritized approach visits candidate records in descending order of their scores because records with high scores have a higher chance of being durable top-$k$ records.
Furthermore, these high-score records
can also serve as a benchmark for future records, enabling a ``blocking mechanism'' to prune candidates. 

Before describing the algorithms, we illustrate the main idea using an example shown in Figure~\ref{fig:blocking}.
Suppose we answer a durable top-3 query with $\tau$ by visiting records in descending order of their scores: $p_1$, $p_2$ and $p_3$, and all three records are durable ones.
$p_1$ has the highest score in the entire query interval, any record that lies in the $\tau$-length time interval $[p_1.t, p_1.t+\tau]$ will be dominated by $p_1$, which we refer to as being ``blocked'' by $p_1$.
Similarly, $p_2$ (the second highest score) and $p_3$ (the third highest score) also block a $\tau$-length interval starting from their arrival times.
The time axis is partitioned into intervals by endpoints of all blocking intervals.
In Figure~\ref{fig:blocking}, the number under each interval shows how many records block this interval. 
Notice the bold red interval, where any record in this interval lies in three blocking intervals after processing $p_1$, $p_2$ and $p_3$.
Since there are already three records with higher score than any record in this interval, it can not have any $\tau$-durable top-3 record, and we can safely remove this time interval from consideration.
As we continue adding blocking intervals, eventually every remaining record in the query interval will be blocked by at least three blocking intervals.
The algorithm can now stop because no more durable top records can be found.
The procedure is straightforwardly applicable to look-ahead version of durability, by simply reversing the direction of blocking intervals.

We describe three algorithms in the following sections. They differ on how the high-score records are found and how the blocking intervals are maintained.

\subsection{Score-Baseline Algorithm}\label{sec:weight-solution:sort}
We start with a baseline method (S-Base) of score-prioritized approach, which sorts records in the query interval in descending order of their scores.
Given $k$, $\tau$ and a query interval $[t_1, t_2]$:
(1) Sort all records in time interval $[t_1-\tau, t_2]$ in descending order of scores.
(2) For each record $p$ in sorted order: If $p.t\in[t_1, t_2]$ and $p$ lies in less than $k$ blocking intervals, add $p$ to answer set;
    Otherwise, continue. 
In any case, add a blocking interval $[p.t, p.t+\tau]$.
    
Since all blocking intervals have the same length $\tau$, we only need to maintain the left endpoints of such intervals (using a balanced binary search tree) to find intersection counts. 
The number of blocking intervals is $O(n)$.
Hence, insertion and query can both be finished in $O(\log n)$ time.
The sorting takes $O(n\log n)$ time so the overall query time complexity of S-Base is $O(n\log n)$.

Next we describe two better algorithms that avoid sorting all records in the query interval.

\subsection{Score-Band Algorithm (Monotone $\pmb{f}$ Only)}\label{sec:weight-solution:sky}
\begin{figure}
    \centering
    \scalebox{.8}{\tikzset{every picture/.style={line width=0.75pt}} 

\begin{tikzpicture}[x=0.75pt,y=0.75pt,yscale=-1,xscale=1]

\draw  (65.3,208.38) -- (247.61,208.38)(65.3,94.2) -- (65.3,209.11) (240.61,203.38) -- (247.61,208.38) -- (240.61,213.38) (60.3,101.2) -- (65.3,94.2) -- (70.3,101.2)  ;
\draw  [color={rgb, 255:red, 0; green, 0; blue, 0 }  ,draw opacity=1 ][fill={rgb, 255:red, 0; green, 0; blue, 0 }  ,fill opacity=1 ] (158.27,193.21) .. controls (158.28,190.53) and (160.45,188.37) .. (163.12,188.38) .. controls (165.79,188.39) and (167.95,190.56) .. (167.94,193.23) .. controls (167.94,195.9) and (165.76,198.06) .. (163.09,198.06) .. controls (160.42,198.05) and (158.26,195.88) .. (158.27,193.21) -- cycle ;
\draw  [color={rgb, 255:red, 0; green, 0; blue, 0 }  ,draw opacity=1 ][fill={rgb, 255:red, 0; green, 0; blue, 0 }  ,fill opacity=1 ] (217.48,155.35) .. controls (217.48,152.68) and (219.66,150.52) .. (222.33,150.53) .. controls (225,150.53) and (227.16,152.7) .. (227.15,155.38) .. controls (227.14,158.05) and (224.97,160.21) .. (222.3,160.2) .. controls (219.63,160.19) and (217.47,158.02) .. (217.48,155.35) -- cycle ;
\draw  [color={rgb, 255:red, 0; green, 0; blue, 0 }  ,draw opacity=1 ][fill={rgb, 255:red, 0; green, 0; blue, 0 }  ,fill opacity=1 ] (136.48,137.35) .. controls (136.48,134.68) and (138.66,132.52) .. (141.33,132.53) .. controls (144,132.53) and (146.16,134.7) .. (146.15,137.38) .. controls (146.14,140.05) and (143.97,142.21) .. (141.3,142.2) .. controls (138.63,142.19) and (136.47,140.02) .. (136.48,137.35) -- cycle ;
\draw  [color={rgb, 255:red, 0; green, 0; blue, 0 }  ,draw opacity=1 ][fill={rgb, 255:red, 0; green, 0; blue, 0 }  ,fill opacity=1 ] (103.27,184.21) .. controls (103.28,181.53) and (105.45,179.37) .. (108.12,179.38) .. controls (110.79,179.39) and (112.95,181.56) .. (112.94,184.23) .. controls (112.94,186.9) and (110.76,189.06) .. (108.09,189.06) .. controls (105.42,189.05) and (103.26,186.88) .. (103.27,184.21) -- cycle ;
\draw  [dash pattern={on 0.84pt off 2.51pt}]  (222.3,160.2) -- (222.09,208.11) ;

\draw  [dash pattern={on 0.84pt off 2.51pt}]  (163.3,198.2) -- (163.09,208.11) ;

\draw  [dash pattern={on 0.84pt off 2.51pt}]  (141.3,142.2) -- (141.12,208.11) ;

\draw  [dash pattern={on 0.84pt off 2.51pt}]  (109.3,187.2) -- (109.12,207.11) ;

\draw  [dash pattern={on 0.84pt off 2.51pt}]  (66.82,137.35) -- (136.48,137.35) ;

\draw  [dash pattern={on 0.84pt off 2.51pt}]  (65.61,155.21) -- (222.3,156.2) ;

\draw  [dash pattern={on 0.84pt off 2.51pt}]  (68.61,185.21) -- (104.27,185.21) ;

\draw  [dash pattern={on 0.84pt off 2.51pt}]  (67.61,194.21) -- (158.27,194.21) ;

\draw [line width=3.75]    (132.39,82.2) -- (132.39,106.2) -- (132.39,123) -- (132.39,165.3) ;

\draw [line width=3.75]    (129.3,167.98) -- (261.3,167.98) ;

\draw  [draw opacity=0][fill={rgb, 255:red, 155; green, 155; blue, 155 }  ,fill opacity=0.67 ] (134.39,82.2) -- (256.3,82.2) -- (256.3,165.3) -- (134.39,165.3) -- cycle ;
\draw [line width=3.75]    (258.3,82.2) -- (258.3,123) -- (258.3,165.3) ;

\draw (281,226) node  [align=left] {Arriving Time};
\draw (60,77) node [rotate=-359.96] [align=left] {Duration};
\draw (105,221) node  [align=left] {$\displaystyle t_{1}$};
\draw (139,221) node  [align=left] {$\displaystyle t_{2}$};
\draw (163,221) node  [align=left] {$\displaystyle t_{3}$};
\draw (224,221) node  [align=left] {$\displaystyle t_{4}$};
\draw (56,181) node  [align=left] {$\displaystyle \tau _{1}$};
\draw (56,135) node  [align=left] {$\displaystyle \tau _{2}$};
\draw (56,194) node  [align=left] {$\displaystyle \tau _{3}$};
\draw (56,153) node  [align=left] {$\displaystyle \tau _{4}$};
\draw (156,120) node  [align=left] {$\displaystyle \tilde{p}_{2}$};
\draw (239,143) node  [align=left] {$\displaystyle \tilde{p}_{4}$};
\draw (179,181) node  [align=left] {$\displaystyle \tilde{p}_{3}$};
\draw (122,169) node  [align=left] {$\displaystyle \tilde{p}_{1}$};

\end{tikzpicture}}
    \caption{Index for $k$-skyband duration.}
    \label{fig:kskyband}
\end{figure}
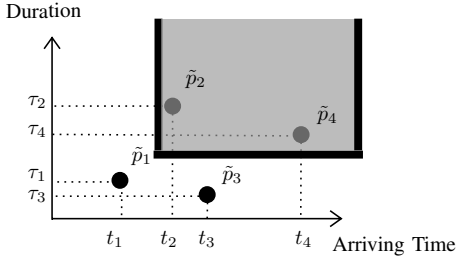
\begin{algorithm}[t]\small
 \KwIn{$P$, $k$, $\tau$, and $\Qinterval$.}
 \KwOut{$\DurTop(k, \Qinterval, \tau)$}
 $S \leftarrow \emptyset$, $\Gamma \leftarrow \emptyset$\;
 Compute $\mathcal{C} \subset P$ by finding durable $k$-skyband set\;
 Sort $\mathcal{C}$ in descending order of scores\;
 \For {$p \in \mathcal{C}$} {
    \If {\text{$p$ lies in $<k$ blocking intervals in   $\Gamma$}} {
        $\pi_{\leq k} \leftarrow Q(k, [p.t-\tau, p.t])$\;
        \uIf {$p \in \pi_{\leq k}$} {
            $S \leftarrow S \cup \{p\}$\;
        }
        \Else {
            \For{$q \in \pi_{\leq k}$ $\wedge$ $q$ not visited before} {
            $\Gamma \leftarrow \Gamma \cup \{[q.t, q.t+\tau]\}$\;
            }
        }
    }
    $\Gamma \leftarrow \Gamma \cup \{[p.t, p.t+\tau]\}$\;
 }
 \KwRet S\;
 \caption{S-Band $(k, \Qinterval, \tau)$\label{algo:ps}}
\end{algorithm}
If we could quickly find a small set of candidate records $\mathcal{C}$, which is guaranteed to be a superset of the answers; i.e., $S \subseteq \mathcal{C}$, then we could get a faster algorithm by only sorting $\mathcal{C}$.
It is well-known that the $k$ records with the highest score, with respect to \emph{any monotone} scoring functions, belong to the $k$-skyband.\footnote{For $\forall p,q, \in P$, $p$ \emph{dominates} $q$ if $p$ is no worse than $q$ in all dimensions, and $p$ is better than $q$ in at least one dimension.
$k$-skyband contains all the points that are dominated by no more than $k-1$ other points. 
Skyline is a special case of $k$-skyband when $k=1$.}
Hence, if a record $p$ is $\tau$-durable for a top-$k$ query (with respect to a monotone $f$), then $p$ must also be $\tau$-durable for the $k$-skyband; i.e., $p$ is in the $k$-skyband for the time interval $[p.t-\tau, p.t]$.
This observation enables us to construct an offline index about each record's duration of belonging to the $k$-skyband, and efficiently produce a superset $\mathcal{C}$ of answers to durable top-$k$ queries.
Note that the score-band algorithm has its limitation, since the $k$-skyband technique only applies to monotone scoring functions.

\mparagraph{Index} Score-Band algorithm needs additional index for finding candidate set $\mathcal{C}$, which we refer to as \emph{durable $k$-skyband}.
Suppose the value of $k$ is known.
For each record $p$, we compute the longest duration $\tau_p$ that $p$ belongs to the $k$-skyband.
Then we map each record $p$ into the ``arrival time - duration'' plane as a two-dimensional point, $\Tilde{p}=(p.t, \tau_p)$.
We then index all such points in the 2D plane using a priority search tree~\cite{de1997computational} (or kd-tree, R-tree in practice).
To answer $\DurTop(k, \Qinterval, \tau)$, we first ask a range query with the $3$-sided rectangle $\Qinterval \times [\tau, +\infty]$.
The set of points that fall into the search region is the superset to actual answers of durable records.
This index can be constructed in $O(n\log n)$ time, has $O(n)$ space and the query time is $O(\card{\mathcal{C}}+\log n)$ in order to get the set $\mathcal{C}$.
Figure~\ref{fig:kskyband} shows an example.
We have four records $p_1, p_2, p_3, p_4$ arriving at $t_1, t_2, t_3, t_4$, whose duration for $k$-skyband is $\tau_1, \tau_2, \tau_3$ and $\tau_4$.
We map them into $\Tilde{p}_1, \Tilde{p}_2, \Tilde{p}_3$ and $\Tilde{p}_4$ according to their arriving time and $k$-skyband duration.
The $3$-sided rectangle $\Qinterval \times [\tau, +\infty]$ is shown as the shaded region.
In this case, $\mathcal{C} = \{p_2, p_4\}$.

%
In general case, notice that we do not know the value of $k$ upfront, i.e., a query has $k$ as a parameter, so we cannot construct only one such index. There are two ways to handle it. 
If we have the guarantee that $k\leq \kappa_0$ for a small number $\kappa_0$ then we can construct $\kappa_0$ such indexes with total space $O(n\kappa_0)$. Otherwise, if $k$ can be any integer in $[1, n]$, we can construct $O(\log n)$ such indexes (priority search trees), one for each $k=2^0, 2^1, \ldots, 2^{\log n}$, so the space is $O(n\log n)$. Given a durable top-$k$ query we first find the number $\bar{k}$ with $k\leq \bar{k}\leq 2k$, and then we use the corresponding index to get the superset $\mathcal{C}$. In this case, $\mathcal{C}$ contains the records that are $\tau$-durable to the $\bar{k}$-skyband, so $S\subseteq \mathcal{C}$.

\begin{figure}
    \centering
    \scalebox{.8}{\tikzset{every picture/.style={line width=0.75pt}} 

\begin{tikzpicture}[x=0.75pt,y=0.75pt,yscale=-1,xscale=1]

\draw    (141.3,108.2) -- (345.3,108.99) ;
\draw [shift={(347.3,109)}, rotate = 180.22] [color={rgb, 255:red, 0; green, 0; blue, 0 }  ][line width=0.75]    (10.93,-3.29) .. controls (6.95,-1.4) and (3.31,-0.3) .. (0,0) .. controls (3.31,0.3) and (6.95,1.4) .. (10.93,3.29)   ;

\draw  [color={rgb, 255:red, 0; green, 0; blue, 0 }  ,draw opacity=1 ][fill={rgb, 255:red, 0; green, 0; blue, 0 }  ,fill opacity=1 ] (285.27,85.21) .. controls (285.28,82.53) and (287.45,80.37) .. (290.12,80.38) .. controls (292.79,80.39) and (294.95,82.56) .. (294.94,85.23) .. controls (294.94,87.9) and (292.76,90.06) .. (290.09,90.06) .. controls (287.42,90.05) and (285.26,87.88) .. (285.27,85.21) -- cycle ;
\draw  [fill={rgb, 255:red, 208; green, 2; blue, 27 }  ,fill opacity=1 ] (251.89,59) -- (262,59) -- (262,69.11) -- (251.89,69.11) -- cycle ;
\draw  [fill={rgb, 255:red, 208; green, 2; blue, 27 }  ,fill opacity=1 ] (204.89,72) -- (215,72) -- (215,82.11) -- (204.89,82.11) -- cycle ;
\draw  [color={rgb, 255:red, 0; green, 0; blue, 0 }  ,draw opacity=1 ][fill={rgb, 255:red, 0; green, 0; blue, 0 }  ,fill opacity=1 ] (313.27,92.21) .. controls (313.28,89.53) and (315.45,87.37) .. (318.12,87.38) .. controls (320.79,87.39) and (322.95,89.56) .. (322.94,92.23) .. controls (322.94,94.9) and (320.76,97.06) .. (318.09,97.06) .. controls (315.42,97.05) and (313.26,94.88) .. (313.27,92.21) -- cycle ;
\draw  [color={rgb, 255:red, 0; green, 0; blue, 0 }  ,draw opacity=1 ][fill={rgb, 255:red, 0; green, 0; blue, 0 }  ,fill opacity=1 ] (160.27,56.21) .. controls (160.28,53.53) and (162.45,51.37) .. (165.12,51.38) .. controls (167.79,51.39) and (169.95,53.56) .. (169.94,56.23) .. controls (169.94,58.9) and (167.76,61.06) .. (165.09,61.06) .. controls (162.42,61.05) and (160.26,58.88) .. (160.27,56.21) -- cycle ;
\draw   (322.3,48.2) .. controls (322.3,43.53) and (319.97,41.2) .. (315.3,41.2) -- (251.3,41.2) .. controls (244.63,41.2) and (241.3,38.87) .. (241.3,34.2) .. controls (241.3,38.87) and (237.97,41.2) .. (231.3,41.2)(234.3,41.2) -- (172.3,41.2) .. controls (167.63,41.2) and (165.3,43.53) .. (165.3,48.2) ;
\draw    (141.3,108.2) -- (141.3,35.2) ;
\draw [shift={(141.3,33.2)}, rotate = 450] [color={rgb, 255:red, 0; green, 0; blue, 0 }  ][line width=0.75]    (10.93,-3.29) .. controls (6.95,-1.4) and (3.31,-0.3) .. (0,0) .. controls (3.31,0.3) and (6.95,1.4) .. (10.93,3.29)   ;

\draw  [dash pattern={on 0.84pt off 2.51pt}]  (165.3,58.2) -- (165.3,109.2) ;

\draw  [dash pattern={on 0.84pt off 2.51pt}]  (209.95,77.05) -- (209.95,110.05) ;

\draw  [dash pattern={on 0.84pt off 2.51pt}]  (257.3,64.2) -- (257.3,109.6) ;

\draw  [dash pattern={on 0.84pt off 2.51pt}]  (289.94,92.23) -- (289.94,108.09) ;

\draw  [dash pattern={on 0.84pt off 2.51pt}]  (318.11,93.2) -- (318.11,107.22) ;

\draw (350,119) node  [align=left] {$\displaystyle t$};
\draw (302,70) node  [align=left] {$\displaystyle p_{4}$};
\draw (273,50) node  [align=left] {$\displaystyle p_{3}$};
\draw (223,63) node [align=left] {$\displaystyle {\textstyle p_{2}}$};
\draw (181,54) node [align=left] {$\displaystyle p_{1}$};
\draw (334,82) node  [align=left] {$\displaystyle p_{5}$};
\draw (252,30) node  [align=left] {$\displaystyle \tau $};
\draw (128,55) node [rotate=-270.65] [align=left] {score};

\end{tikzpicture}}
    \caption{Durability checks in S-Band and S-Hop.}
    \label{fig:false-check}
\end{figure}
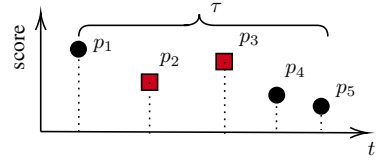

\mparagraph{Query Algorithm} We refer to this score-prioritized approach using durable $k$-skyband candidates as Score-Band algorithm, or S-Band.
Full algorithm is sketched in Algorithm~\ref{algo:ps} and described below.
Given $k, \Qinterval, \tau$, we first retrieve the candidate set $\mathcal{C}$ using the durable $k$-skyband index as shown above.
Then we sort $\mathcal{C}$ and visit records in descending order of their scores.
For each record $p$ we visit, we first check the number of blocking intervals that $p$ lies.
If $p$ lies in less than $k$ blocking intervals, it is a promising candidate and we run a top-$k$ query on time interval $[p.t - \tau, p.t]$ for durability check.
If $p$ is indeed $\tau$-durable, we add $p$ to answer set. 
Otherwise, we need to add a blocking interval for each record returned by the top-$k$ query (if we have not done so yet), since they all have higher scores than $p$.
On the other hand, if $p$ already lies in at least $k$ blocking intervals, we can simply skip it. 
In the end, we add the blocking interval $[p.t, p.t+\tau]$ for $p$.

We can see that S-Band works similarly to S-Base.
The only difference is that for a record that is blocked less than $k$ times, we still have to execute a top-$k$ query to check whether the record is $\tau$-durable (Line 6).
This step of durability check is necessary.
Though some records are guaranteed to be non-durable (i.e., not captured by $\mathcal{C}$ with durable $k$-skyband), they can still block other records (with lower scores) to be durable ones.
Consider a concrete example in Figure~\ref{fig:false-check} where black dots represent candidate records in $\mathcal{C}$ and red squares represent records that are not in $\mathcal{C}$.
S-Band would only visit $p_1, p_4$ and $p_5$.
At the time we visit $p_4$, there is only one blocking interval (introduced by $p_1$).
However, $p_2$ and $p_3$ actually have higher scores than $p_4$.
By running a durability check query on $p_4$, we can discover these missing records and add corresponding blocking intervals (Line 10-11) for better pruning power in future steps.  

\mparagraph{Complexity} The query time complexity of S-Band can be decomposed into three parts: 1) a range search query to find candidate set $\mathcal{C}$; 2) sort $\mathcal{C}$ according to their scores; 3) find durable records from sorted $\mathcal{C}$ sequentially.
Summing up the above, the overall query time complexity of S-Band is $O\big(\card{\mathcal{C}}(q(n) + k)\log n\big)$, assuming that a top-$k$ query can be answered in $O(q(n) + k)$ time.
In the worst case $\card{\mathcal{C}}=O(n)$ since all points can lie in the $k$-skyband. In Section~\ref{sec:answer-size} we show that using the probabilistic model in~\cite{bentley1977average} (where the coordinates of the points are randomly assigned) the expected size of $\mathcal{C}$ is $O(k\big\lceil\tfrac{\card{\Qinterval}}{\tau}\big\rceil\log^{d-1}\tau)$. Due to the blocking mechanism, in practice we expect that the number of top-$k$ queries will be smaller. 
However, notice that we always need to sort all records in $\mathcal{C}$ which might make S-Band much slower due to the size of $\mathcal{C}$ that increases (in expectation) exponentially on the dimension $d$.

\subsection{Score-Hop Algorithm}\label{sec:weight-solution:top1}
The data reduction strategy of S-Band offers adequate benefits for improving the overall running time on datasets in low dimensions ($\leq 5$).
However, the query overhead on searching and sorting candidate records becomes a huge burden on high-dimensional data, as it is well-known that the size of $k$-skyband tends to explode (or equivalently, records in high-dimensional space tends to stay in $k$-skyband for a longer duration) in high-dimensional space.
Furthermore, S-Band requires additional index and only applies to monotone scoring functions.
To overcome the drawbacks of S-Base and S-Band, we propose another approach that does not require sorting and has better worst case guarantee.
The main idea is that there is no need to sort records in advance; we can find the record with the next highest score one by one as we find durable records.
With the help of blocking mechanism, we can skip certain time intervals when we find the next highest score record, despite the fact that there might be some high-score records in such intervals. 
This procedure has an analogy to the Time-Hop algorithm, since we effectively skip certain records while we traverse records in descending order of their scores, as we taking a hop in the score-domain.

\begin{algorithm}[t]\small
 \KwIn{$P$, $k$, $\tau$, and $\Qinterval : [a,b]$.}
 \KwOut{$\DurTop(k, \Qinterval, \tau)$}
 $H \leftarrow \emptyset$, $S \leftarrow \emptyset$, $\Gamma \leftarrow \emptyset$\;
 \For {$[l_i, r_i] :$ disjoint $\tau$-length intervals in $\Qinterval$} {
 $M_i \leftarrow Q(\wtv, k, [l_r, r_i])$\;
 $H.$push($M_i$.pop())\;
 }
 \While {$H \neq \emptyset$} {
 $p \leftarrow  H.$pop(), and let $p\in M_j$\;
 \uIf {$p$ lies in $<k$ blocking intervals in  $\Gamma$} {
 $\pi_{\leq k} \leftarrow Q(\wtv, k, [p.t-\tau, p.t])$\;
 \uIf {$p \in \pi_{\leq k}$} {
    $S \leftarrow S \cup \{p\}$\;
 }
 \Else {
    \For{$q \in \pi_{\leq k}$ $\wedge$ $q$ not visited before} {
            $\Gamma \leftarrow \Gamma \cup \{[q.t, q.t+\tau]\}$\;
            }
 }
 $M_j^- \leftarrow Q(k, [l_j, p.t-1])$\;
 $M_j^+ \leftarrow Q(k, [p.t+1, r_j])$\;
 $H.$push($M_j^-$.top()), $H.$push($M_j^+$.top())\; 
 }
 \ElseIf{$M_j \neq \emptyset$}{
 $H.$push($M_j.$pop())\;
 }
 \If{$p$ not visited before}{
 $\Gamma \leftarrow \Gamma \cup \{[p.t, p.t+\tau]\}$\;
 }
 }
 \KwRet S\;
 \caption{S-Hop $(k, \Qinterval, \tau)$\label{algo:psi}}
\end{algorithm}

\mparagraph{Query Algorithm} We refer to this solution as Score-Hop algorithm, or S-Hop.
The main idea of the algorithm is straightforward.
In each iteration, we find the record with the maximum score among the records that lie in less than $k$ blocking intervals. Let $p$ be such a record. We run a durable top-k query so if $p$ is a $\tau$-durable record we add it in $S$. If $p$ is not a $\tau$-durable record, we add a blocking interval for each record returned by the durable top-k query
(if they have not been added before). In the end, we add the blocking interval $[p.t, p.t+\tau]$ and we continue with the next record with the highest score.
The actual implementation of the algorithm is more subtle, to guarantee a fast query time as described below; pseudo-code is provided in Algorithm~\ref{algo:psi}.
Given a query interval $\Qinterval=[a,b]$, we partition the interval into a set of disjoint $\tau$-length sub-intervals: $[a, a+\tau),[a+\tau, a+2\tau), \ldots, [a+\big\lfloor\tfrac{\card{\Qinterval}}{\tau}\big\rfloor\tau, b]$.
Let $[l_i, r_i]$ be the $i$-th sub-interval, and in each interval we find the $k$ records
\footnote{\label{note}As a practical note, we notice that finding the top-1 record (instead of top-$k$) in each time interval can be more efficient in most real-life datasets.} 
with the highest score, denoted $M_i$.
We construct a max-heap $H$ over all the top-1 records from all sub-intervals.
Besides that, each node in $H$ also keeps the original interval $[l_i, r_i]$ and the set $M_i$ associated with the record.
We repeat the following until $H$ is empty.
We take and pop the top record from $H$. 
Let $p$ be that record originated from $M_j$.
Then $p$ will be processed in the following two cases:
1) If $p$ lies in at least $k$ blocking intervals, we update $H$ by pushing the next top record in $M_j$ (if there is any). 
2) If $p$ lies in less than $k$ blocking intervals, we update $H$ as follows.
Assume that $[l_j, r_j]$ is the corresponding sub-interval of $M_j$ (or $p$).
We first split $[l_j, r_j]$ into two non-empty intervals $[l_j, p.t-1]$ and $[p.t+1, r_j]$.
Then, run a top-$k$ query on $[l_j, p.t-1]$ to get a new top-$k$ set $M_j^-$.
Similarly, get another new set $M_j^+$ from $[p.t+1, r_j]$.
We replace the old set $M_j$ with $M_j^-$ and $M_j^+$, along with its corresponding interval $[l_j, p.t-1]$ and $[p.t+1, r_j]$, respectively.
Finally, we update $H$ by pushing the current top records from $M_j^-$ and $M_j^+$ into the heap.
In the end, we add the blocking interval from record $p$ (if it is the first time we visited $p$).
Figure~\ref{fig:s-hop} illustrates the main procedure of S-Hop on how to find next record with highest score. 
It is worth mentioning that the hopping movement happens at Line 18: 
we effectively skip certain intervals by not updating the max-heap and stop asking top-$k$ queries on its sub-intervals.
\footnote{In practice, we make sure that when we ask $k$ top-$1$ queries in an interval we remove it from the max-heap.}
\begin{figure}
    \centering
    \scalebox{.8}{\tikzset{every picture/.style={line width=0.75pt}} 

\begin{tikzpicture}[x=0.75pt,y=0.75pt,yscale=-1,xscale=1]

\draw   (259.11,36.11) -- (388.61,67.11) -- (130.61,67.11) -- cycle ;
\draw    (110.3,116.2) -- (417.61,116.2) ;
\draw [shift={(419.61,116.2)}, rotate = 180] [color={rgb, 255:red, 0; green, 0; blue, 0 }  ][line width=0.75]    (10.93,-3.29) .. controls (6.95,-1.4) and (3.31,-0.3) .. (0,0) .. controls (3.31,0.3) and (6.95,1.4) .. (10.93,3.29)   ;

\draw [line width=2.25]    (133,107) -- (133,128.11) ;

\draw [line width=2.25]    (215,108) -- (215,129.11) ;

\draw [line width=2.25]    (297,108) -- (297,129.11) ;

\draw [line width=2.25]    (379,108) -- (379,129.11) ;

\draw  [fill={rgb, 255:red, 247; green, 6; blue, 6 }  ,fill opacity=1 ] (254.6,35.6) .. controls (254.6,33.11) and (256.62,31.09) .. (259.11,31.09) .. controls (261.6,31.09) and (263.62,33.11) .. (263.62,35.6) .. controls (263.62,38.09) and (261.6,40.11) .. (259.11,40.11) .. controls (256.62,40.11) and (254.6,38.09) .. (254.6,35.6) -- cycle ;
\draw  [fill={rgb, 255:red, 245; green, 5; blue, 5 }  ,fill opacity=1 ] (350.6,116.6) .. controls (350.6,114.11) and (352.62,112.09) .. (355.11,112.09) .. controls (357.6,112.09) and (359.62,114.11) .. (359.62,116.6) .. controls (359.62,119.09) and (357.6,121.11) .. (355.11,121.11) .. controls (352.62,121.11) and (350.6,119.09) .. (350.6,116.6) -- cycle ;
\draw   (354.61,105.11) .. controls (354.61,100.44) and (352.28,98.11) .. (347.61,98.11) -- (336.39,98.11) .. controls (329.72,98.11) and (326.39,95.78) .. (326.39,91.11) .. controls (326.39,95.78) and (323.06,98.11) .. (316.39,98.11)(319.39,98.11) -- (302.61,98.11) .. controls (297.94,98.11) and (295.61,100.44) .. (295.61,105.11) ;
\draw   (380.61,105.11) .. controls (380.61,101.68) and (378.89,99.97) .. (375.46,99.97) -- (375.46,99.97) .. controls (370.56,99.97) and (368.11,98.25) .. (368.11,94.82) .. controls (368.11,98.25) and (365.66,99.97) .. (360.76,99.97)(362.96,99.97) -- (360.76,99.97) .. controls (357.33,99.97) and (355.61,101.68) .. (355.61,105.11) ;

\draw (420,127) node  [align=left] {$\displaystyle t$};
\draw (153,37) node  [align=left] {Max Heap $\displaystyle H$};
\draw (172,136) node  [align=left] {$\displaystyle M_{i}$};
\draw (340,136) node  [align=left] {$\displaystyle M_{j}$};
\draw (259,103) node  [align=left] {...};
\draw (307,17) node  [align=left] {Heap Top: $\displaystyle p\ \in M_{j}$};
\draw (362,125) node  [align=left] {$\displaystyle p$};
\draw (371,81) node  [align=left] {$\displaystyle M^{+}_{j}$};
\draw (331,80) node  [align=left] {$\displaystyle M^{-\ }_{j}$};
\draw (134,140) node  [align=left] {$\displaystyle l_{i}$};
\draw (217,139) node  [align=left] {$\displaystyle r_{i}$};
\draw (298,140) node  [align=left] {$\displaystyle l_{j}$};
\draw (381,139) node  [align=left] {$\displaystyle r_{j}$};

\end{tikzpicture}}
    \caption{Illustration of Score-Hop algorithm on finding next record with highest score (if $p$ lies in less than $k$ blocking intervals).}
    \label{fig:s-hop}
\end{figure}
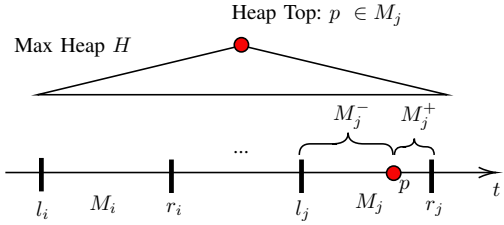

Compared to S-Band, S-Hop does not have a strong dependency on the dimension of the data (only the running time of the top-$k$ queries depends on the dimension) and makes better use of the blocking mechanism.
In the end, we only find and process high-score records as we need instead of acquiring a full sorted order of records in advance, which leads to better worst case theoretical guarantees and faster query time.
Experimental results in Section~\ref{sec:expr} demonstrate that S-Hop can be 1 to 2 orders of magnitude faster than S-Band on high-dimensional ($\geq$ 10) datasets.

\mparagraph{Correctness} 
The following lemma proves the correctness of S-Hop.
\begin{lemma}\label{lemma:psi:correctness}
Given $k$, $\Qinterval$ and $\tau$,
the Score-Hop algorithm returns the correct answer for durable top-$k$ query. 
\end{lemma}
\begin{proofSk}
Let $S^*$ be the $\tau$-durable records in $\Qinterval$.
We show that $S\subseteq S^*$ and $S^*\subseteq S$.
The algorithm always checks by running a top-$k$ query if a record should be in the solution (line 8 of Algorithm~\ref{algo:psi}) so $S\subseteq S^*$.

Next we prove $S^*\subseteq S$.
The algorithm visits the records in descending (score) order so it is not possible that a record $p\in S^*$ lies in at least $k$ blocking intervals before the algorithm visits $p$.
We also need to prove that the algorithm does not miss any durable record in a sub-interval $[l_j, r_j]$ that corresponds to an empty $M_j$.
If $\card{P([l_j, r_j])}\leq k$ then the result follows.
Otherwise, we argue using induction that each time when the algorithm finds a record $p$ in $M_j$ that is contained in at least $k$ blocking intervals,  any timestamp in the sub-interval $[l_j, p.t]$ lies in at least $k$ blocking intervals.
Hence, if $M_j$ is empty, any timestamp in $[l_j, r_j]$ lies in at least $k$ blocking intervals and no other durable records are in $[l_j, r_j]$.
\end{proofSk}

\subsection{Complexity Analysis of S-Hop}\label{sec:weight-solution:analysis}
The query complexity analysis of S-Hop is non-trivial and needs more care.
There are three main sub-procedures in S-Hop: find next highest score record, top-$k$ queries for durability check and blocking mechanism.
As presented above, the first two components both rely on multiple top-$k$ queries.
We first show a worst-case guarantee on the total number of top-$k$ queries called in the algorithm.
Please refer to Appendix~\ref{appendix:proof2} for full proof.
\begin{lemma}\label{lemma:ps-bound}
The total number of top-$k$ queries performed by the Score-Hop algorithm is $O(\card{S} + k\big\lceil\frac{\card{\Qinterval}}{\tau}\big\rceil)$.
\end{lemma}
\begin{proofSk}
As we had in the proof of Lemma~\ref{lemma:asw-bound} we need to bound the number of false checks. Let $p$ be a false check and let $p'$ be the record with the largest timestamp in $Q(k, [p.t-\tau, p.t])$. We say that $p$ is assigned to $p'$. If $p'.t<a$, where $a$ is the timestamp such that $\Qinterval=[a, b]$, then we assign $p$ to $a$. We first show that at the moment that we find the false check $p$ the corresponding record $p'$ can only have one of the following three properties: i) it lies in at least $k$ blocking intervals, ii) $p'\in S$ and it lies in at most $k-1$ blocking intervals, iii) $p'=a$.
If $p'$ has property ii) then $p$ is a type-1 false check. Otherwise, $p$ is a type-2 false check.

We first bound the number of type-1 false checks. Notice that after a type-1 false check $p$ is assigned to $p'$ then all timestamps in the sub-interval $[p'.t, p.t]$ lie in at least $k$ records. So if another false check $q$ later in the algorithm is assigned to $p'$, again, then $q$ can only be a type-2 false check. Hence, the type-1 false checks are bounded by $O(\card{S})$.
In order to bound the type-2 false checks we assume a window $\rho$ of length $\tau$ in $\Qinterval$. We make the following key observation: At the moment that we find a type-2 false check $p$, it lies in at most $k-1$ blocking intervals while $p'$ lies in at least $k$ blocking intervals, so there should be a blocking interval $[l, r]$, where its right endpoint lies between $p'.t$ and $p.t$, i.e., $p'.t\leq r\leq p.t$.
(Notice that if $p'=a$ is assigned more than once then it also lies in at least $k$ blocking intervals.)
Using this observation along with other properties of the false checks we can show that after finding $k$ type-2 false checks in $\rho$, each timestamp in $\rho$ will lie in at least $k$ blocking intervals. Hence, the algorithm will not run any other top-$k$ query in $\rho$. Since there are $\big\lceil\frac{\card{\Qinterval}}{\tau}\big\rceil$ disjoint $\tau$-length sub-intervals in $\Qinterval$ we can bound the total number of type-2 false check by $O(k\big\lceil\frac{\card{\Qinterval}}{\tau}\big\rceil)$.
Overall, the number of false checks along with the durable records in $\Qinterval$ is $O(\card{S} + k\big\lceil\frac{\card{\Qinterval}}{\tau}\big\rceil)$.
\end{proofSk}
The lemma above also shows that the number of different sets $M_j$ that are created by the algorithm is 
$O(\card{S} + k\big\lceil\tfrac{\card{\Qinterval}}{\tau}\big\rceil)$. For each set we can visit at most $k$ records so in total the algorithm may visit $O(k(\card{S} + k\big\lceil\tfrac{\card{\Qinterval}}{\tau}\big\rceil))$ records \footnote{We note that the algorithm may visit some records, that lie in at least $k$ blocking intervals, more than once. The upper bound $O(k(\card{S} + k\big\lceil\frac{\card{\Qinterval}}{\tau}\big\rceil))$ counts all the times that the algorithm visits a record. We can modify the algorithm so that it does not visit the same record twice but that would make the description of the algorithm more complicated without decreasing the overall asymptotic complexity.}.
Each top(), or pop() procedure takes $O(\log n)$ time so in total we need $O(k(\card{S} + k\big\lceil\tfrac{\card{\Qinterval}}{\tau}\big\rceil)\log n)$ to visit these records.
Furthermore, recall that we need $O(\log n)$ time to check if a record lies in at least $k$ blocking intervals and $O(\log n)$ time to insert a blocking interval (using a binary search tree)
so we also spend $O(k(\card{S} + k\big\lceil\frac{\card{\Qinterval}}{\tau}\big\rceil)\log n)$ time for the blocking mechanism.
Notice that this running time is dominated by the time to answer $O(\card{S} + k\big\lceil\frac{\card{\Qinterval}}{\tau}\big\rceil)$ top-$k$ queries, so S-Hop answers a durable preference top-$k$ query in 
$O\big((\card{S} + k\big\lceil\tfrac{\card{\Qinterval}}{\tau}\big\rceil)(q(n) + k)\log n\big)$
time (with an efficient top-$k$ query procedure in $O(q(n) + k)$).
Similarly to T-Hop our index for S-Hop has near linear space and supports updates in polylogarithmic time for preference queries.

As it turns out, hopping in time-domain (T-Hop) and in score-domain (S-Hop) gives us the same complexity bound.
But in practice, S-Hop is more conservative in asking preference top-$k$ queries compared to T-Hop, due to the candidate pruning brought by blocking mechanism.
This makes S-Hop run faster than T-Hop when the top-$k$ query itself is expensive; i.e., a larger $k$ or on high-dimensional datasets.
\section{Expected Complexity}\label{sec:answer-size}
In the previous sections we presented two types of algorithms (time-prioritized and score-prioritized) to answer durable top-$k$ queries with the same worst-case guarantee on their query time.
In particular we showed that their query times depend on $k\big\lceil\frac{\card{\Qinterval}}{\tau}\big\rceil$ and $\card{S}$.
In this section, we go beyond the worst-case analysis and analyze their performance in a more ``expected'' sense.
Most importantly, we show in Section~\ref{sec:answer-size:random} that the expected size of $\card{S}$ is roughly $k\big\lceil\frac{\card{\Qinterval}}{\tau}\big\rceil$ if the scores of data records are drawn randomly from an \emph{arbitrary} distribution (which can be picked by a powerful adversary with the advance knowledge of the query parameters).
This result essentially establishes that, under this model, our best algorithms are in a sense optimal because their complexity is expected to be linear in the output size.
Secondly, in Section~\ref{sec:answer-size:kSkyband}, we study the expected complexity of Score-Band algorithm by bounding the expected size of $\tau$-durable $k$-skyband candidate set $\mathcal{C}$ using the same probabilistic model used in \cite{bentley1977average}.

\subsection{Expected Answer Size}\label{sec:answer-size:random}
Consider a set of $n$ records $P$ with $p_i.t=i$, for $p_i\in P$.
We analyze the expected size of a query output when the scores of records are assigned in a semi-random manner, where the data values can be arbitrarily chosen and then they are assigned in a random order to the records.
More formally, we consider a \emph{random permutation model} (RPM).
Let $\mathbf{X}=x_1 < x_2 < \ldots < x_n$ be a sequence of $n$ arbitrary non-negative numbers chosen by an adversary, and let $\sigma$ be a permutation of $\{1, \ldots, n\}$. We set $f(p_i)=x_{\sigma(i)}$, i.e., the score of record $p_i$ is $x_{\sigma(i)}$, where $\sigma(i)$ is the image of $i$ under $\sigma$.
As argued in~\cite{agarwal2018range}, the random permutation model is more general than the model in which all scores are drawn from an arbitrary unknown distribution, so our result holds for this model as well. 
The random permutation model has been widely used in a rich variety of domains and considered as a standard for complexity analysis; i.e., online algorithms~\cite{goel2008online, mahdian2011online, mehta2005adwords}, discrete geometry~\cite{agarwal2018union, agarwal2014union, har2015complexity}, and query processing~\cite{agarwal2018range}.
Our main result is the following.

\begin{lemma}\label{lemma:expected-answer-size}
In the random permutation model, given $k, \tau$ and $\Qinterval$, we have $\displaystyle \Expec{\cardin{S}}=k\tfrac{\card{\Qinterval}}{\tau+1}$.
\end{lemma}

\begin{proof}
For a record $p_i\in P(\Qinterval)$, let $X_i$ be the random variable, which is $1$ if $\pnt_i$ is a $\tau$-durable record, and $0$ otherwise. 
Thus, $\Expec{\cardin{S}} = \Expec{\sum_{i}X_i}.$
Using the linearity of expectation, $\Expec{\sum_{i}X_i}=\sum_{i}\Expec{X_i}=\sum_{i}\Prob{X_i=1}.$
    
Thus our goal is to compute $\Prob{X_i=1}$ : the probability that there are less than $k$ records in $[p_i.t-\tau, p_i.t)$ with score larger than $f(p_i)$. 
Let $P_i^\tau=\{\pnt_{i-\tau}, \ldots, \pnt_{i-1}\}$.
For a subset $Q\subset P_i^\tau$, let $A_Q$ be the binary random variable, which is $1$ if
all records in $Q$ have score greater than $f(\pnt_i)$ and all records in $\overline{Q}=P_i^\tau\setminus Q$ have score less than $f(\pnt_i)$.
We have 
\begin{equation}\label{eq:p_xi}
\Prob{X_i=1}=\sum_{l=0}^{k-1} \sum_{Q\subset P_i^\tau, \cardin{Q}=l}\Prob{A_Q}.%
\end{equation}

We estimate $\Prob{A_Q}$ as follows.
Let $V\subset \mathbf{X}$ with $\cardin{V}=\tau+1$.
We first bound the conditional probability $\Prob{A_Q\mid V}$ such that the records in $P_i^\tau\cup \{p_i\}$ are assigned scores from $V$.
We consider all possible permutations of $V$ and count only those cases where the records in $Q$ have larger value than $f(\pnt_i)$, and the records in $\overline{Q}$ have values less than the value of $f(\pnt_i)$. 
Notice that the permutations that satisfy this property must assign the first $l$ largest values of $V$ to $Q$, then the $(l+1)$-th largest value to $\pnt_i$ and the rest $\tau-l$ smaller values of $V$ to $\overline{Q}$.
Under such assignment, any permutations of values in $Q$ and $\overline{Q}$ are valid cases.
Hence, the number of valid permutations are $l!(\tau-l)!$, while the number of all possible permutations of $V$ are $(\tau+1)!$. 
We have
\begin{equation}\label{eq:p_q_v}
\Prob{A_Q\mid V}=\frac{l!(\tau-l)!}{(\tau+1)!}=\frac{1}{\tau+1}\frac{1}{{\tau \choose l}}.%
\end{equation}
Since (\ref{eq:p_q_v}) holds for all $V$, $\displaystyle\Prob{A_Q}=\frac{1}{\tau+1}\frac{1}{{\tau \choose l}}$.
Substituting this in (\ref{eq:p_xi}), we obtain
\begin{equation}\label{eq:bound}\resizebox{0.9\hsize}{!}{$%
    \Prob{X_i=1} =\sum_{l=0}^{k-1} {\tau \choose l}\frac{1}{\tau+1}\frac{1}{{\tau \choose l}}=\sum_{l=0}^{k-1}\frac{1}{\tau+1}=\frac{k}{\tau+1}%
$}
\end{equation}
Finally,
\begin{equation}
    \Expec{\card{S}}= \sum_{i}\Prob{X_i=1} =k\frac{\card{\Qinterval}}{\tau+1}.
\end{equation}
\end{proof}
Combining Lemma~\ref{lemma:expected-answer-size} with
the analysis of Sections~\ref{sec:time-solution:analysis} and~\ref{sec:weight-solution:analysis}, we conclude that in a random permutation model the \emph{expected} query time complexity of both Time-Hop and Score-Hop algorithms is
$O(\card{S}(q(n)+k)\log n)$, or equivalently $O\big(k\big\lceil\frac{\card{\Qinterval}}{\tau}\big\rceil(q(n)+k)\log n\big)$, where $O(q(n) + k)$ reflects the time complexity of answering a top-$k$ query.
In Section~\ref{sec:expr}, our experimental results on real and synthetic datasets both confirm this finding.

\subsection{Expected size of durable $k$-skyband}\label{sec:answer-size:kSkyband}
In this subsection we bound the expected size of $\tau$-durable $k$-skyband records, denoted by $\mathcal{C}$, from Section~\ref{sec:weight-solution:sky} in a probabilistic model similar to the previous case.
Recall that the size of $\mathcal{C}$ affects the running time of the S-Band algorithm.

Let $P=\{\pnt_1, \ldots, \pnt_n\}$ with $\pnt_i.t=i$.
We use the same random model as in \cite{bentley1977average} where (the attributes of) records are randomly generated.
The following lemma bounds the expected size of $\mathcal{C}$ (See Appendix~\ref{appendix:proof3}).
\begin{lemma}\label{lemma:expected-Skyband-size}
In the random model as in~\cite{bentley1977average},
given $k, \tau$ and $\Qinterval$, we have $\Expec{\card{\mathcal{C}}}=O(k\frac{\card{\Qinterval}}{\tau}\log^{d-1}\tau)$.
\end{lemma}
Combining Lemma~\ref{lemma:expected-Skyband-size}, the analysis of Section~\ref{sec:weight-solution:sky} and an efficient top-$k$ query procedure runs in $O(q(n) + k)$ time, the \emph{expected} query time complexity of Score-Band algorithm is 
$O\big(k\big\lceil\frac{\card{\Qinterval}}{\tau}\big\rceil(q(n)+k)\log n\log^{d-1}\tau\big)$.
It shows that the expected complexity of Score-Band algorithm can be higher than Time-Hop or Score-Hop algorithm by a factor of at most $\log^{d-1} \tau$.
Experimental results in Section~\ref{sec:expr} also confirm this finding as we vary the data dimensionalities.
The curse of dimensionality makes Score-Band algorithm perform worse even compared to other simple baselines.
Again, Time-Hop and Score-Hop are both generally applicable to arbitrary user-specified scoring functions, while Score-Band only works for monotone functions.
\section{Experiments}\label{sec:expr}

\subsection{Experiment Setup}\label{sec:expr:setup}
\begin{table}[]\small
    \centering
    \caption{Dataset summary}
    \begin{tabular}{|c|c|c|}
    \hline
        Dataset & Dimensionality & Size (\# records)  \\\hline
        \textbf{NBA-X} & \small{1,2,3,5} & \small{1M} \\\hline
        \textbf{Network-X} & \small{2,3,5,10,20,30,37} & \small{5M} \\\hline
        \textbf{Syn-X} & 2 & \small{1M,2M,5M,10M,20M,50M} \\\hline
    \end{tabular}
    \label{tab:data}
\end{table}
\mparagraph{Datasets}\label{sec:expr:dataset}
We use two real-life datasets and some synthetic ones, as summarized in Table~\ref{tab:data} and described below:

\textbf{NBA}\footnote{NBA datasets were collected from \url{https://www.basketball-reference.com/}} contains the performance of \emph{each} NBA player in \emph{each} game from 1983 to 2019, with in total $\sim 1$ million individual performance records on 15 numeric attributes.
Records are naturally organized by date and time, and we break ties (e.g., performances of different players in the same game) arbitrarily. 
We choose some subsets of 15 attributes to create datasets with different dimensions collectively referred to as \textbf{NBA-X}:
\textbf{NBA-1} selects only \emph{3-point-made};
\textbf{NBA-2} captures the \emph{points} and \emph{assists};
\textbf{NBA-3} chooses \emph{points}, \emph{assists}, \emph{rebounds};
\textbf{NBA-5} includes five dimensions: \emph{points}, \emph{assists}, \emph{rebounds}, \emph{steals} and \emph{blocks}.
    
\textbf{Network}\footnote{\url{https://kdd.ics.uci.edu/databases/kddcup99/kddcup99.html}} is the dataset from KDD Cup 1999.
This dataset contains $\sim5$ million records with 37 numeric attributes that describe network connections to a machine, including \emph{connection duration}, \emph{packet size}, etc.
The query in this case utilizes a scoring function that weighs a variety of numerical attributes to rank connections in order to identify unusual and potentially malicious ones.
Records have unique timestamps and are ordered by these timestamps.
Since these attributes have different measurement units, we scale the value of each dimension using MinMax normalization.
To study the impact of data dimensionalities on query efficiency, we choose the first 2, 3, 5, 10, 20, 30 and 37 attributes from the full dimensions to create 7 different datasets collectively referred to as \textbf{Network-X}, where $\mathbf{X}$ represents the dimensionality of the dataset.
\begin{figure}[t]\vspace*{-3ex}
    \centering
    \subfloat[IND]{\includegraphics[width=0.24\textwidth]{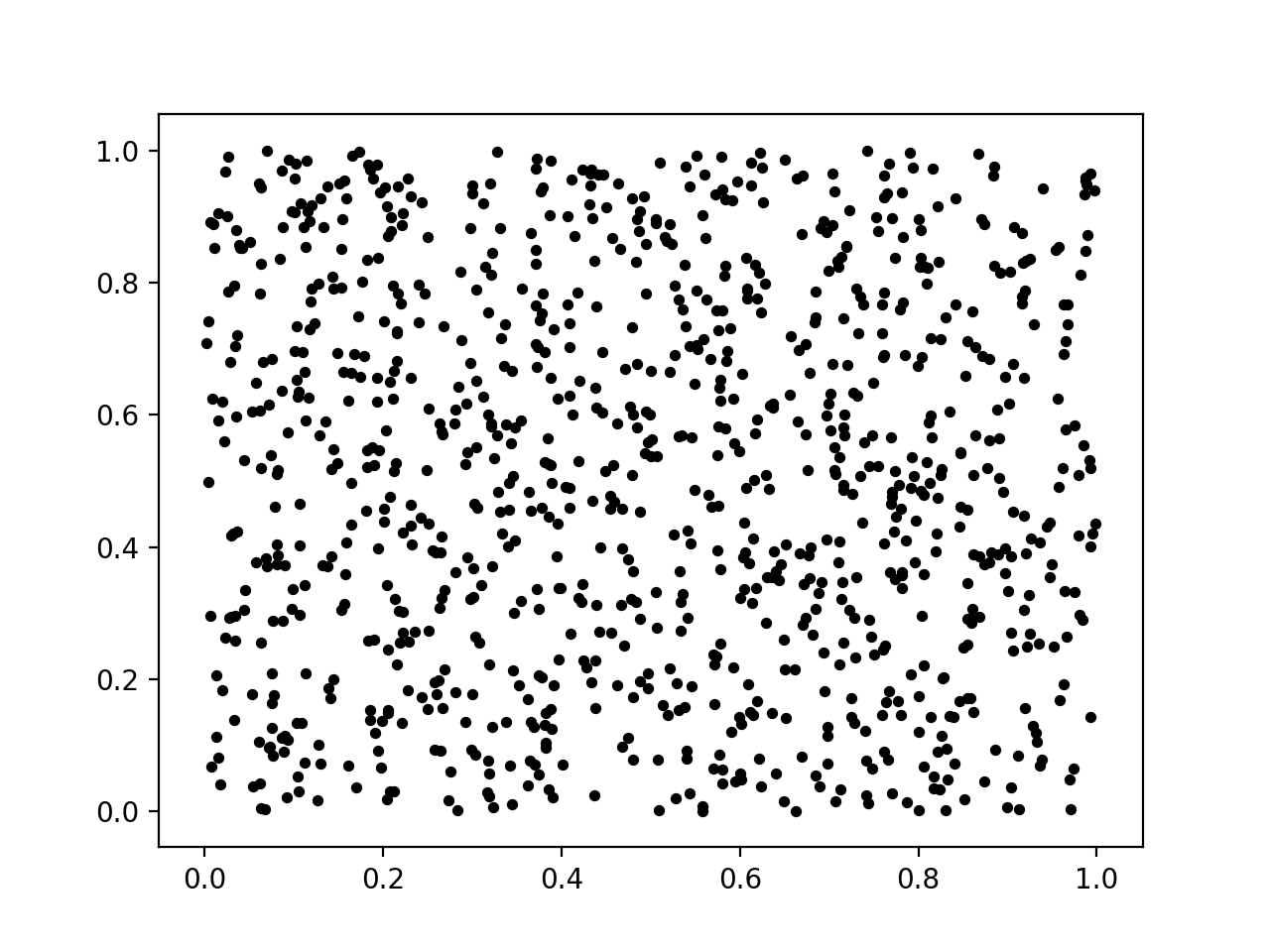}}
    \subfloat[ANTI]{\includegraphics[width=0.24\textwidth]{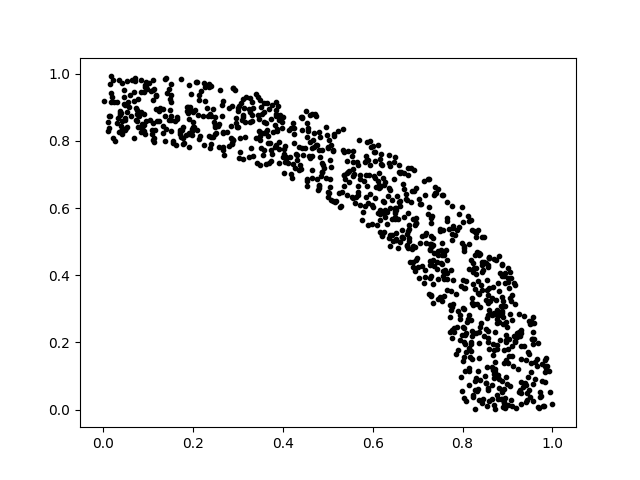}}
    \caption{Value distributions for synthetic dataset}
    \label{fig:syn}
\end{figure}   

\textbf{Syn} is a synthetic two-dimensional dataset that is used for scalability test on proposed solutions.
We generate Syn with independent (IND) and anti-correlated (ANTI) data distributed in a 2D unit square.
For IND data, the attribute values of each tuple are generated independently, following a uniform distribution. 
ANTI data are drawn from the portion inside the positive orthant of an annulus centered at the origin with outer radius 1 and inner radius 0.8, representing an environment that most of the records gather in $k$-skyband.
Figure~\ref{fig:syn} illustrates the sample value distributions of IND and ANTI.
The full size of Syn is 50 million and each data point has an unique arriving time.
We further choose several subsets of Syn with 1, 2, 5, 10 and 20 millions of records. 
The set of synthetic datasets are collectively referred to as \textbf{Syn-X}, where $\mathbf{X}$ represents data size.

\mparagraph{Query Parameters}
\begin{table}[t]\small
    \centering
    \caption{Query Parameters (default value in bold)}
    \begin{tabular}{|c|c|}
    \hline
    Parameter  & Range  \\\hline
    $k$ & 5, \textbf{10}, 15, 20, 25, 30, 35, 40, 45, 50\\\hline
    $\tau$ & 1\%, 5\%, 10\%, 15\%, \textbf{20\%}, 25\%, 30\%, 40\%, 50\%\\\hline
    $\card{\Qinterval}$ & 10\%, 20\%, 30\%, 40\%, \textbf{50\%}, 60\%, 70\%\ 80\%\\\hline
    $d$ & 1, \textbf{2}, 3, 5, 10, 20, 30, 37 \\\hline
    \end{tabular}
    \label{tab:parameter}
\end{table}
Table~\ref{tab:parameter} summaries the query parameters under investigation, along with their ranges and default values.
Among these, the query interval length $\card{\Qinterval}$ and the durability $\tau$ is measured as percentage of dataset size $n$.
When varying query interval length, we always fix the right endpoint of the interval to be the most recent timestamp in dataset and only move the left endpoint.

\begin{figure*}
\begin{minipage}[t]{0.48\textwidth}
\centering
    \subfloat[Performance on NBA-2 as $\tau$ varies.]{\includegraphics[width=0.5\textwidth]{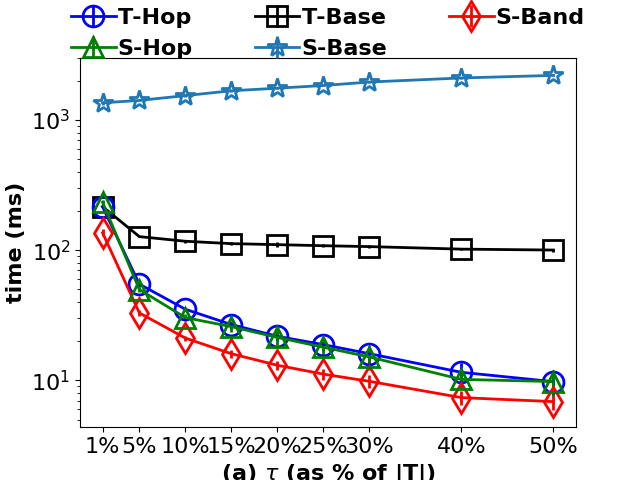}%
    \includegraphics[width=0.5\textwidth]{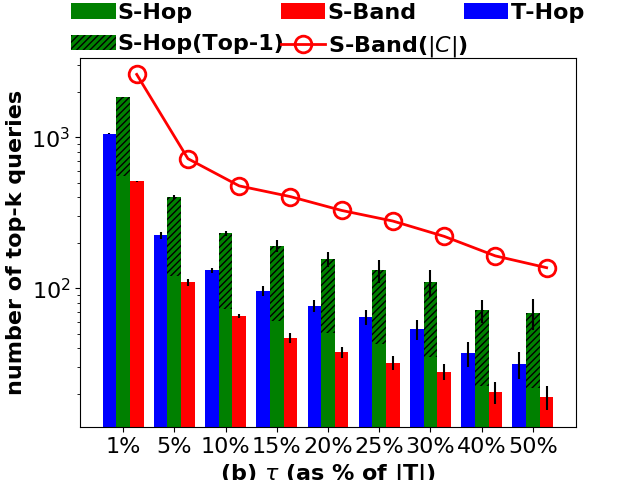}}\\
    \subfloat[Performance on Network-2 as $\tau$ varies.]{\includegraphics[width=0.5\textwidth]{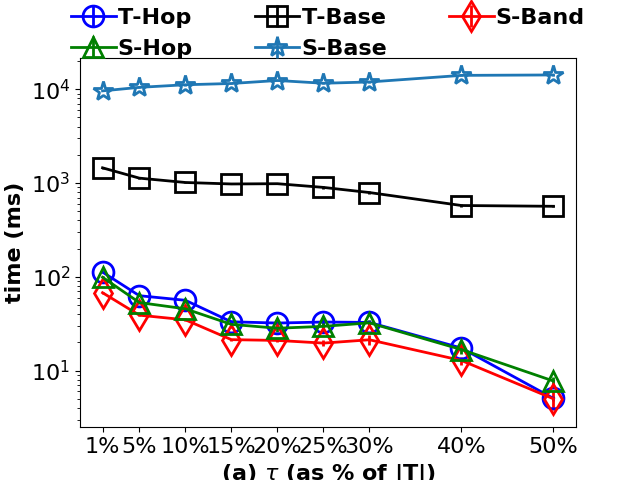}%
    \includegraphics[width=0.5\textwidth]{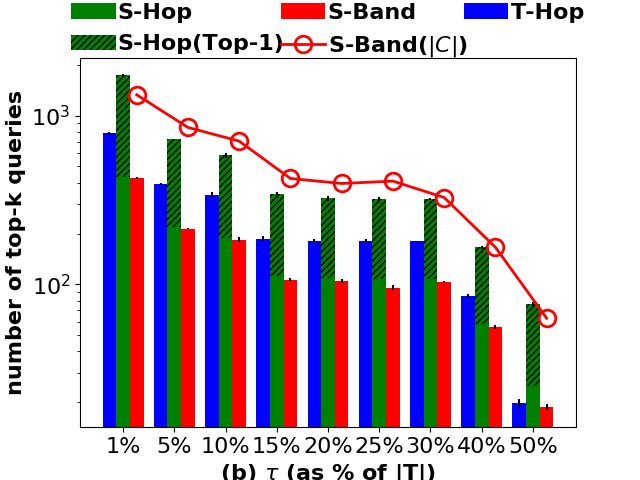}}
    \caption{Performance comparison as $\tau$ varies.}
    \label{fig:tau}
\end{minipage}
\begin{minipage}[t]{0.48\textwidth}
\centering
    \subfloat[Performance on NBA-2 as $k$ varies.]{\includegraphics[width=0.5\textwidth]{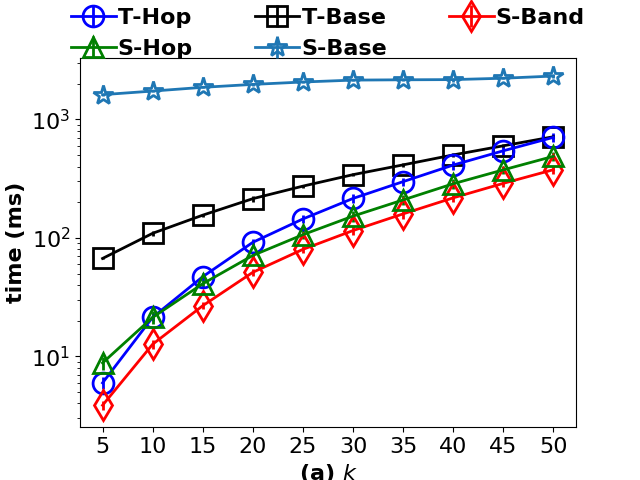}%
    \includegraphics[width=0.5\textwidth]{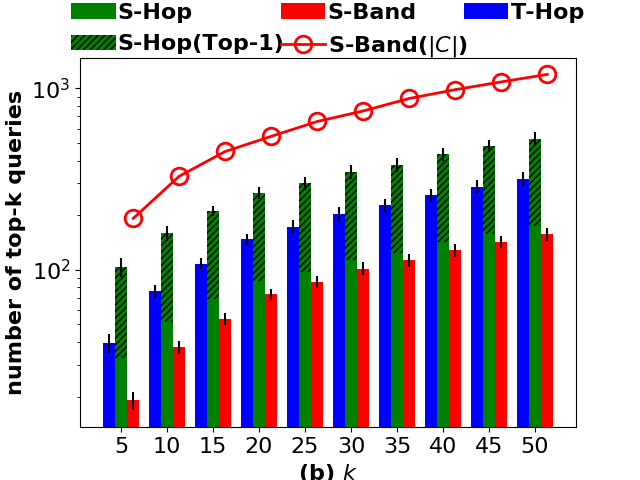}}\\
    \subfloat[Performance on Network-2 as $k$ varies.]{\includegraphics[width=0.5\textwidth]{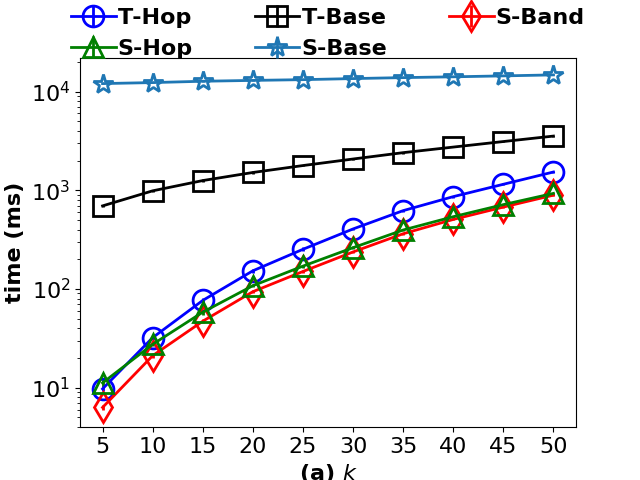}%
    \includegraphics[width=0.5\textwidth]{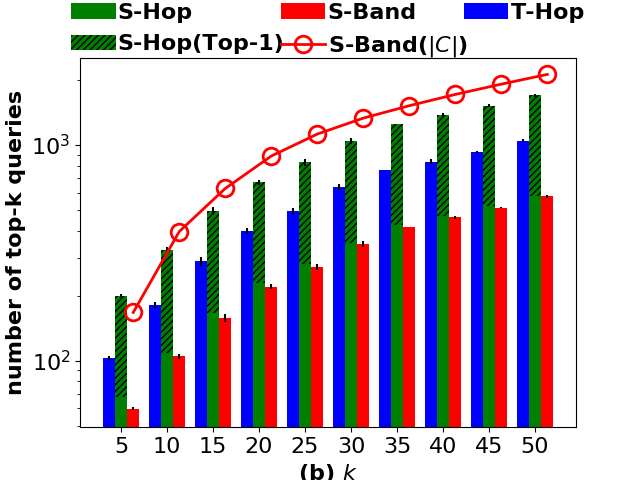}}
    \caption{Performance comparison as $k$ varies.}
    \label{fig:k}
\end{minipage}
\vspace{-1em}
\end{figure*}

\vspace{-1em}
\mparagraph{Implementations \& Evaluation Metric}
To make the discussions concrete and concise, we choose a linear and monotone preference scoring function throughout the experimental section in the simple form: $\score(p) =  \sum_{i=1}^{d} \wtv_i \cdot p.x_i$, where $\wtv$ is a user-specified preference vector and $\wtv_i$ is the (non-negative) weight for $i$-th attribute of a record.
At query time, user need to specify $\wtv$ as one of the input parameters.
Since the focus of this paper is not to develop the best possible index for top-$k$ queries $Q_\wtv(k, W)$, 
our implementation of the top-$k$ building block simply adopts a tree index (on the time domain of $P$), and answer $Q_\wtv(k, W)$ in a straightforward top-down manner with a branch-and-bound method. 
More specifically, each tree node stores the \emph{skyline} of all records that it contains.
The skyline helps us quickly identify the maximum score of each node under \emph{any} preference vector $\wtv$.
Then, to answer $Q_\wtv(k, W)$, it is sufficient to use at most $k$ nodes (that are contained by time window $W$) with the highest scores according to $\wtv$. 
This index offers adequate performance in our experiments,
but it can certainly be replaced by more sophisticated index with
better worst-case guarantees, without affecting the rest of our proposed solution.

Using the building block of top-$k$ queries described above, we further implement \textbf{T-Base} (Section~\ref{sec:time-solution:sw}), \textbf{T-Hop} (Section~\ref{sec:time-solution:esw}), \textbf{S-Base} (Section~\ref{sec:weight-solution:sort}), \textbf{S-Band} (Section~\ref{sec:weight-solution:sky}) and \textbf{S-Hop} (Section~\ref{sec:weight-solution:top1}).
Performance of various methods are evaluated using the following two metrics: number of top-$k$ queries and overall query time (in millisecond).
For each query parameter setting, we run the query 100 times with 100 different randomly generated preference vectors, and report the average with standard deviation.

All methods were implemented in C++, and all experiments were performed on a Linux machine with two Intel Xeon E5-2640 v4 2.4GHz processor with 256GB of memory.

\subsection{Algorithm Evaluations}\label{sec:expr:runtime}
According to the theoretical analysis of our algorithms in previous sections, the query efficiency depends on the length of durability window $\tau$, the value of $k$, the length of query interval $\Qinterval$, the data dimensionality $d$ and the data size $n$.
For fair evaluation and comparison of algorithm efficiency, we designed a set of variable-controlling experiments such that each time we only vary one query parameter of interest and fix the others to default values.

\vspace{-0.3em}
\mparagraph{Comparison of Algorithms when Varying $\tau$}\label{sec:expr:tau}
In Figure~\ref{fig:tau}, we investigate the performance of all durable top-$k$ solutions, as we vary durability $\tau$.
Figure~\ref{fig:tau}-1-(a) shows the query efficiency comparison on NBA-2.
The sorting based solution S-Base is the slowest, as it requires fully sorting all records in the time interval of length $\card{\Qinterval} + \tau$.
T-Base is faster than S-Base and mostly independent of $\tau$.
All the rest solutions, T-Hop, S-Hop and S-Band, become more efficient as we increase $\tau$, or equivalently, when query is more selective.
This finding confirms our analysis in Section~\ref{sec:answer-size} that the query efficiency bounds of Hop-based solutions and S-Band both depend on the answer size, which is $O(k\frac{\card{\Qinterval}}{\tau})$.
T-Hop and S-Hop nearly perform the same, while S-Band can be slightly faster.
When the query is highly selective ($\tau$ is half of the length of entire time domain), they are 1-2 orders of magnitude faster compared to T-Base and S-Base, respectively.
Similar trends can be seen in Figure~\ref{fig:tau}-2-(a), where we test algorithms on a larger dataset Network-2.
The only difference is that baseline solutions (T-Base and S-Base) are more expensive and the efficiency difference between baseline solutions and T-Hop/S-Hop/S-Band is even larger (up to 3 orders of magnitude).

Next, we take a closer look at T-Hop, S-Hop and S-Band in Figure~\ref{fig:tau}-1-(b), which compares the number of top-$k$ queries needed for these three advanced algorithms.
For S-Hop, the total number of top-$k$ queries is decomposed into two parts: top-$k$ queries for durability check (unshaded region of a green bar) and top-$k$ queries for finding the next highest score record (shaded region).
For S-Band, we also plot the size of durable $k$-skyband candidate set $C$ on top the figure as red circled line, reflecting the overhead cost of sorting $C$ for S-Band.
Now it is clear that the main reason why T-Hop/S-Hop/S-Band becomes faster when $\tau$ is large is that fewer top-$k$ queries are needed.
A more selective query with larger $\tau$ also makes the candidate set $C$ of S-Band smaller, demonstrating the effectiveness of using durable $k$-skyband to identify promising candidates.
On the other hand, we can see that S-Hop and S-Band ask fewer top-$k$ queries than T-Hop, demonstrating the pruning power of blocking mechanism in score-prioritized solutions.
This figure also explains why S-Band runs slightly faster than S-Hop and T-Hop on NBA-2 in this case, as S-Band requires the least number of top-$k$ queries and the overhead cost on sorting candidate set $C$ is relatively small on two-dimensional data.
Again, similar trends can be found in Figure~\ref{fig:tau}-2-(b).

\begin{figure}[t]
    \centering
    \subfloat[Performance on NBA-2 as $\card{\Qinterval}$ varies.]{\includegraphics[width=0.24\textwidth]{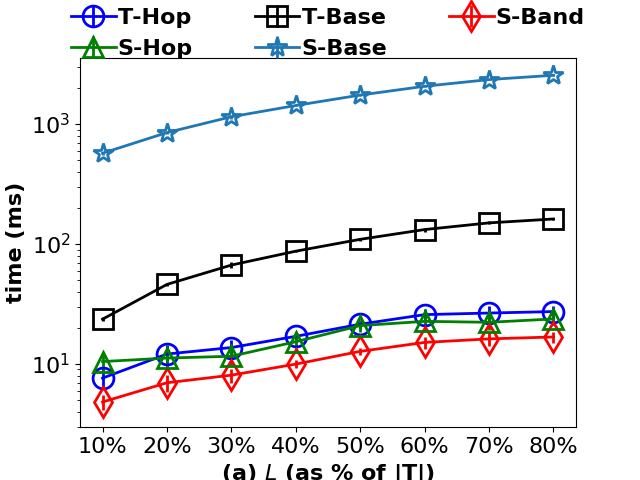}%
    \includegraphics[width=0.24\textwidth]{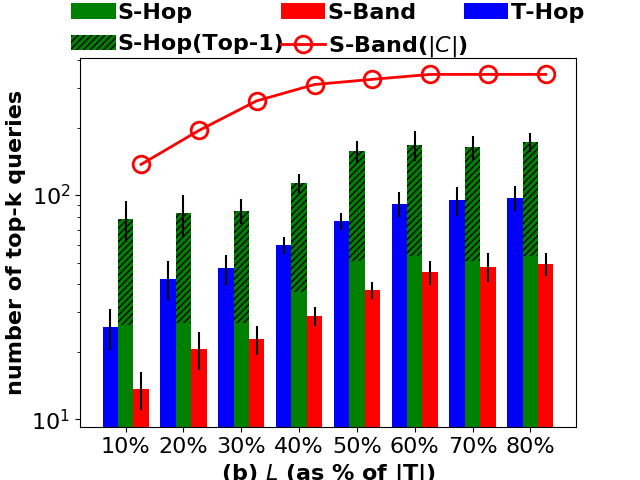}}\\
    \subfloat[Performance on Network-2 as $\card{\Qinterval}$ varies.]{\includegraphics[width=0.24\textwidth]{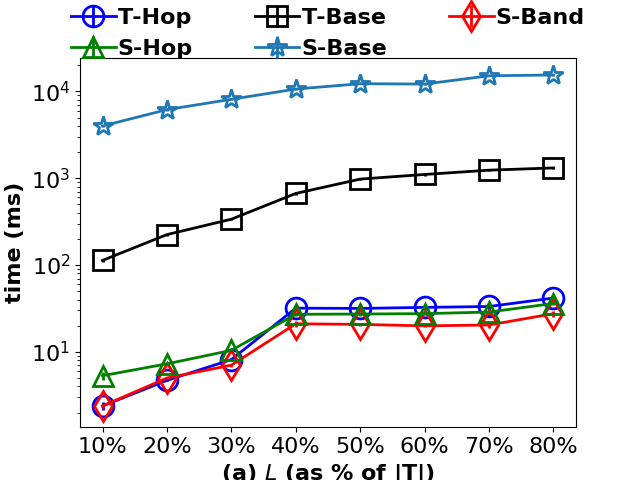}%
    \includegraphics[width=0.24\textwidth]{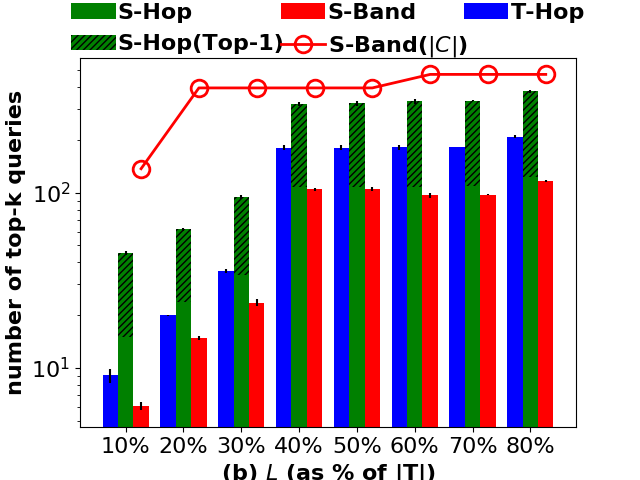}}
    \caption{Performance comparison as $\card{\Qinterval}$ varies.}
    \label{fig:L}
    \vspace{-1em}
\end{figure}

\vspace{-0.3em}
\mparagraph{Comparison of Algorithms when Varying $k$}\label{sec:expr:k}
Next, we study the effect of $k$ on efficiency.
Results are shown in Figure~\ref{fig:k}.
When we increase $k$, not only need we ask more top-$k$ queries (see Figure~\ref{fig:k}-1-(b) and Figure~\ref{fig:k}-2-(b)), but a top-$k$ query itself also becomes more expensive.
Thus in both Figure~\ref{fig:k}-1-(a) and Figure~\ref{fig:k}-2-(a), we can see that all algorithms (except S-Base) are slower when $k$ is larger.
Especially when $k$ reaches 50, top-$k$ computations become the dominant factor on overall efficiency, and the differences among the various algorithms diminish.
Still, S-Band and S-Hop have slight advantages over T-Hop on larger $k$, as they use blocking mechanism to prune candidate records and are more conservative in asking expensive top-$k$ queries.

\vspace{-0.3em}
\mparagraph{Comparison of Algorithms when Varying $\card{\Qinterval}$}\label{sec:expr:L}
In Figure~\ref{fig:L}, we compare the performance of proposed algorithms as we vary the query interval length $\card{\Qinterval}$.
In terms of efficiency, Figure~\ref{fig:L}-1-(a) and Figure~\ref{fig:L}-2-(a) show that T-Hop/S-Hop/S-Band is much faster than baseline solutions T-Base and S-Base, especially on the large dataset Network-2.
On the other hand, we also find that our proposed algorithms scale better with $\card{\Qinterval}$ than with $k$ (recall Figure~\ref{fig:k}).
The reason is that the time complexities of T-Hop/S-Hop and S-Band are quadratic in $k$ but only linear on $\card{\Qinterval}$ (recall Lemma~\ref{lemma:expected-answer-size} and Lemma~\ref{lemma:expected-Skyband-size}).
In terms of number of top-$k$ queries, in Figure~\ref{fig:L}-1-(b) and Figure~\ref{fig:L}-2-(b), it is not surprising to see that all proposed solutions ask more top-$k$ queries as $\card{\Qinterval}$ increases.
The relative performance of various algorithms is consistent with previous experiments where we varied $\tau$ or $k$.   

\vspace{-0.3em}
\mparagraph{Comparison of Algorithms when Varying $d$}\label{sec:expr:d}
\begin{figure}[t]
    \vspace{-1em}
    \centering
    \subfloat[]{\includegraphics[width=0.24\textwidth]{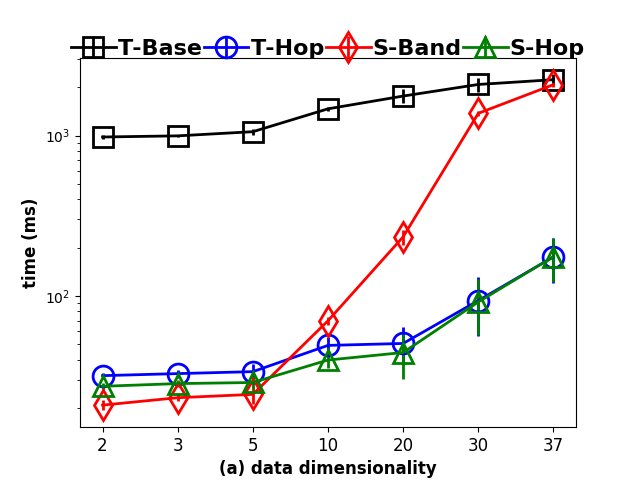}}
    \subfloat[]{\includegraphics[width=0.24\textwidth]{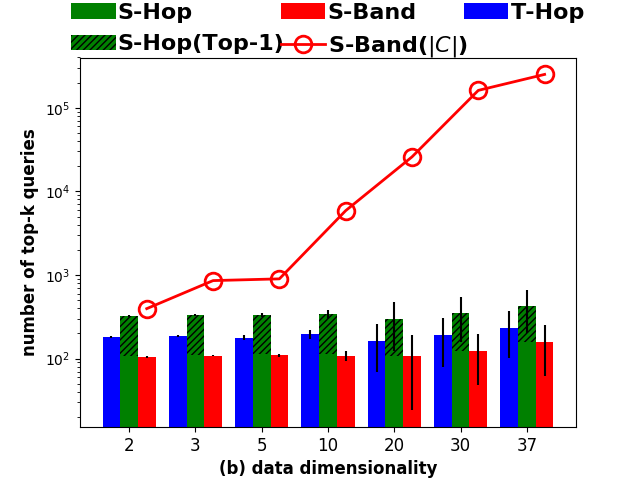}}
    \caption{Performance comparison on Network-X as $d$ varies.}
    \label{fig:d}
    \vspace{-2em}
\end{figure}
In this section, we study the effect of data dimensionality $d$ on algorithm performances.
Since the sorting-based S-Base is clearly inferior to other algorithms, here we only test T-Base, T-Hop, S-Band and S-Hop on Network-X with varying dimensions.
Results are shown in Figure~\ref{fig:d}.
Let us first take a look on Figure~\ref{fig:d}-2.
We can see that the number of top-$k$ queries for all proposed algorithms stays stable as we increase dimensionality.
This finding again confirms our theoretical analysis that the number of top-$k$ queries (or, answer size) depends only on $k\frac{\card{\Qinterval}}{\tau}$ and is independent of dimensionality $d$.
On the other hand, we can see that the size of candidate set $C$ for S-Band rockets in high dimensions, and can be up to 4 orders of magnitude larger than the size of actual promising records.
The sorting overhead on such huge candidate sets is already too big.
Then, let us go back to Figure~\ref{fig:d}-1.
The query time of T-Base, T-Hop and S-Hop slowly increases as we increase dimensionality, because top-$k$ queries on high-dimensions become more expensive, yet they ask roughly the same number of top-$k$ queries regardless of dimensionality.
While S-Band still performs well on low-dimensional data (less than 5 dimensions), in higher dimension S-Band becomes dramatically worse, even taking as much time as T-Base on Network-37.

\vspace{-0.3em}
\mparagraph{Scalability}\label{sec:expr:scale}
\begin{figure}
    \centering
    \subfloat[IND]{\includegraphics[width=0.24\textwidth]{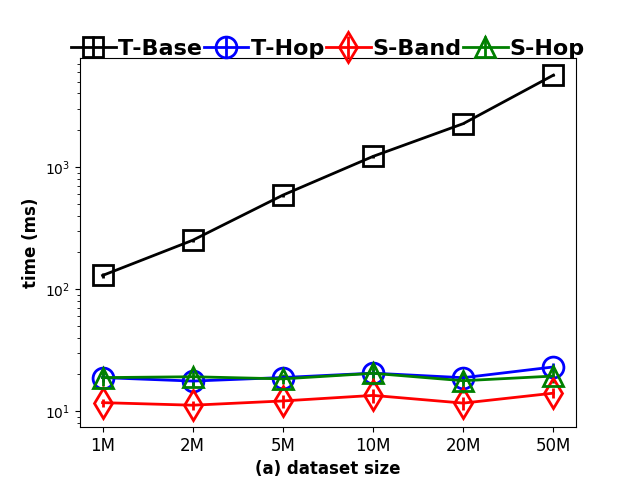}%
    \includegraphics[width=0.24\textwidth]{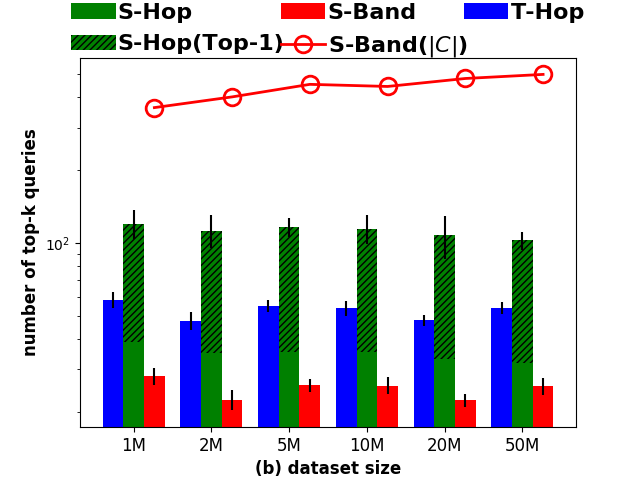}}\\
    \subfloat[ANTI]{\includegraphics[width=0.24\textwidth]{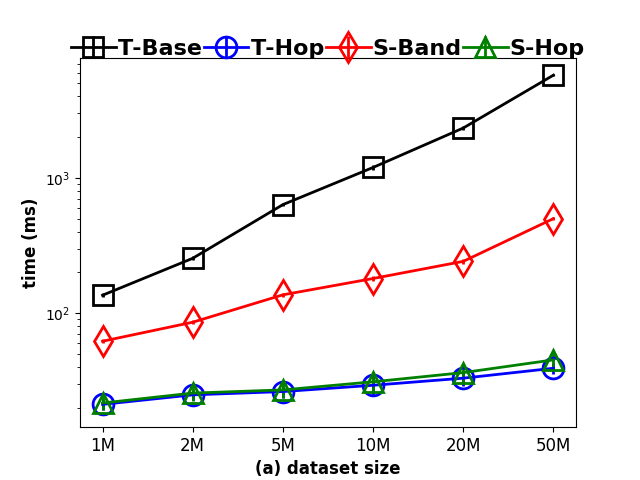}%
    \includegraphics[width=0.24\textwidth]{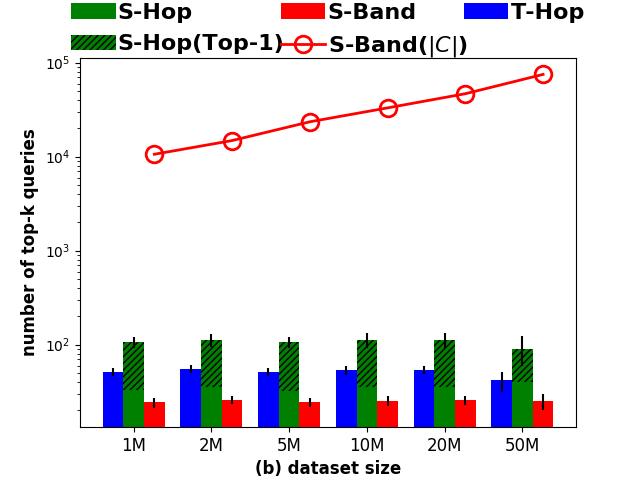}}
    \caption{Scalability test on IND and ANTI Syn-X.}
    \label{fig:scalability}
    \vspace{-1em}
\end{figure}
Finally, we use the two-dimensional synthetic dataset Syn-X to test the scalability of the proposed algorithms as we vary the input size from 1 million to 50 million.
Figure~\ref{fig:scalability} summarizes the results.
As the input size increases, we also increase the query interval length proportionally (so it remains at a fixed percentage of the data size).
As shown in Figure~\ref{fig:scalability}-1, we can see that T-Hop, S-Hop and S-Band scale well on large IND datasets, and S-Band again performs slightly better than T-Hop and S-Hop.
The running time of S-Base increases on larger datasets simply because we are also making the query interval longer.
Figure~\ref{fig:scalability}-1-(b) further illustrates that the total number of top-$k$ queries asked by different algorithms is also independent from the data size.
A larger dataset only makes top-$k$ queries more expensive.
Although the size of candidate set $\card{C}$ increases on larger IND datasets, its growth rate here is much lower than its growth rate when varying dimensionality $d$ in Figure~\ref{fig:d}.
Overall, on IND synthetic data, $\card{C}$ is only about 4-5 times bigger than the actual answer size, which will not incur a big sorting overhead for S-Band.
However, the situation is much different for ANTI Syn-X.
As shown in Figure~\ref{fig:scalability}-2, in terms of query efficiency, T-Hop and S-Hop still scale well, but S-Band now becomes much more expensive because of the data distribution of ANTI.
Most records in ANTI data would gather in $k$-skyband, resulting in $C$ up to 3 orders of magnitude larger than the actual answer size (see Figure~\ref{fig:scalability}-2-(b)), which hurts the performance of S-Band.
The efficiency of S-Band has a strong dependency on the candidate set $C$, or more generally, the data distribution.
In contrast, the performance of T-Hop and S-Hop in this case is nearly independent of both size and distribution of data; it is only linear to the answer size.

\mparagraph{Query Time Distribution over Different Real Datasets}\label{sec:expr:dist}
\begin{figure}[t]
    \centering\vspace*{-2ex}
    \includegraphics[width=0.35\textwidth]{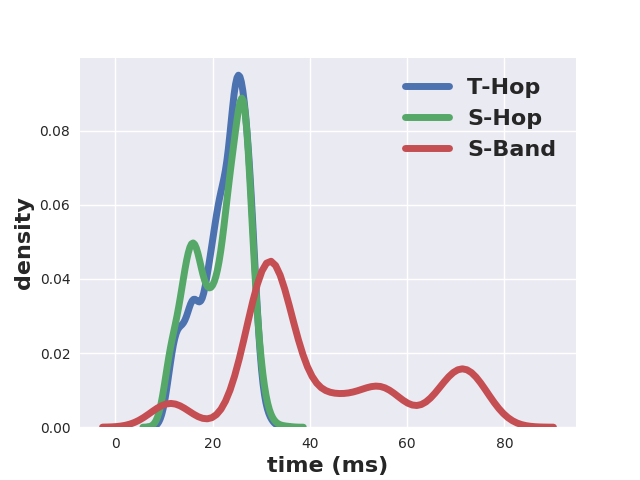}
    \caption{Runtime distribution on 5d NBA data.}
    \label{fig:distribution}
\end{figure}
Figure~\ref{fig:scalability} already clearly illustrates the performance difference of S-Band on IND and ANTI synthetic data, demonstrating the effect of data distributions on S-Band's query efficiency.
Here, we further compare T-Hop, S-Hop and S-Band on real data, and study how data distributions would influence their performance in practice.
We use NBA as the main data source, and select 20 combinations of 5 dimensions randomly chosen out of the 15 attributes, e.g., (points, assits, rebounds, steals, blocks), (points, assits, steals, blocks, 3-pointers-made), etc.
These resulting 20 datasets have the same dimensionality (5) but exhibit different distributions.
We run queries with default settings on each dataset, and plot the running time distribution for all datasets.
Results are shown in Figure~\ref{fig:distribution}.
We can see that S-Band takes longer time on average, and also has a wide span on query time.
This finding again confirms that S-Band is highly sensitive to underlying data distributions.
In contrast, running times of T-Hop and S-Hop are centered in narrower value ranges, showing their robustness to data distributions and further demonstrating their advantages over S-Band on real data.

In sum, we conclude that the Hop-based algorithms, T-Hop and S-Hop, are the best solutions for answering durable preference top-$k$ queries.
They scale well on large datasets as well as to high dimensions, and most importantly, their query time complexity is proportional to the answer size.
This property makes T-Hop and S-Hop run even faster when the query is highly selective; i.e., smaller $k$ or larger $\tau$, which tend to be the more practical and meaningful query settings that people would use in real-life applications.
While S-Band is also a reasonable approach, its performance depends highly on the data characteristics (faring poorly in high dimensions and for certain distributions).
S-Band also requires additional offline indexing for finding durable $k$-skyband candidates. 
Overall, as demonstrated by experiments on both real and synthetic data, efficiency and robustness of Hop-based solutions make them more attractive solutions.
Even on very large and high-dimensional datasets, T-Hop/S-Hop only need less than a second to return durable top records for any given preference, which enables interactive data exploration.

\subsection{DBMS-Based Implementations}
\begin{table}
    \centering
    \caption{Query time (in seconds) comparison on NBA-2 when varying $\tau$. PostgreSQL backend.}
    \begin{tabular}{|c|c|c|c|c|c|}
    \hline
    $\tau$ (as \% of $|T|$) & 10\% & 20\% & 30\% & 40\% & 50\%\\\hline
    T-Hop & 0.46 & 0.28 & 0.18 & 0.12 & 0.1  \\\hline
    T-Base & 2.2 & 1.9 & 1.8 & 1.7 & 1.7  \\\hline
    \end{tabular}
    \label{tab:dbms_by_tau}
\end{table}
\begin{table}
    \centering
    \caption{Query time (in seconds) comparison on NBA-2 when varying $|I|$. PostgreSQL backend.}
    \begin{tabular}{|c|c|c|c|c|c|}
    \hline
    $|I|$ (as \% of $|T|$) & 10\% & 20\% & 30\% & 40\% & 50\%\\\hline
    T-Hop & 0.1 & 0.16 & 0.17 & 0.2 & 0.26 \\\hline
    T-Base & 0.46 & 0.93 & 1.3 & 1.6 & 2  \\\hline
    \end{tabular}
    \label{tab:dbms_by_L}
\end{table}
\begin{table}
    \centering
    \caption{Query time (in seconds) comparison on different datasets. Dataset size (measured by DBMS storage size) is shown in parentheses.
    PostgreSQL backend.
    }
    \begin{tabular}{|c|c|c|c|}
    \hline
    Dataset & NBA-2 (\textbf{0.05 G}) & Syn-IND (\textbf{30 G}) & Syn-ANTI (\textbf{30 G})\\\hline
    T-Hop & 0.28 & 1.9 &  2.3 \\\hline
    T-Base & 1.9 & 773 & 787  \\\hline
    \end{tabular}
    \label{tab:dbms_by_size}
\end{table}
\textcolor{black}{
To demonstrate the generality of proposed solutions and its possibility of integrating into a DBMS, we further test the algorithms utilizing PostgreSQL~\cite{postgresql} as the backend DBMS.
More specially, we load the datasets NBA-2, Syn-500M (IND) and Syn-500M (ANTI) into PostgreSQL tables.
The table schema consist of numeric attributes of the records and an additional column representing arriving time instant.
For algorithm implementations, we code T-Hop and T-Base as stored procedures using PL/Python with PostgreSQL's native support operators.\footnote{The other proposed solution, S-Hop, requires a more delicate query procedure and data structures (recall Algorithm~\ref{algo:psi}). Hence it is more suitable to implement S-Hop as a wrapper function outside the DBMS.}
Besides data tables, we also create corresponding index tables to support efficient top-$k$ records retrieval.
The index table is similar to the tree-based index as we used for previous experiments, providing sufficient data reduction for answering range top-$k$ queries.
Again, the top-$k$ module can be replaced by more sophisticated indexes with better performance, without affecting the rest of our solution.
}

\textcolor{black}{
Tables~\ref{tab:dbms_by_tau} and \ref{tab:dbms_by_L} show the results of testing T-Hop and T-Base on the smaller NBA-2 dataset with the same query setting as before, varying durability $\tau$ and query interval length $|I|$ to compare query efficiencies.
Similar conclusions can be drawn here.
T-Base always pays linear cost (continuous sliding windows) to visit all records in the query interval. 
Thus, the running time is linear to $|I|$ (Table~\ref{tab:dbms_by_L}), and nearly independent of $\tau$ (Table~\ref{tab:dbms_by_tau}).
In comparison, T-Hop's complexity is linear to the answer size, which makes it run faster as query becomes more selective (smaller $|I|$ or larger $\tau$).
Overall, T-Hop is at least 10$\times$ faster than T-Base.
}

\textcolor{black}{
In Table~\ref{tab:dbms_by_size}, we increase the dataset size up to 500M records, which takes around 30 Gigabytes of disk space in PostgreSQL.
Running default queries in such cases, we can see that T-Hop is more than 100$\times$ faster than T-Base, bringing down the query time from nearly 12 minutes to just 2 seconds.
T-Hop also apparently scales well on large datasets, since the complexity is mostly linear to the answer size.
The query time increase solely comes from the more expensive top-$k$ module.
On the contrary, the continuous sliding-window nature of T-Base makes it prohibitively slow when dealing with large amounts of temporal data.
}
\vspace{-0.5em}
\subsection{Summary of Experiments}
In sum, we conclude that the Hop-based algorithms, T-Hop and S-Hop, are the best solutions for answering durable preference top-$k$ queries.
They scale well on large datasets as well as to high dimensions, and most importantly, their query time complexity is proportional to the answer size.
This property makes T-Hop and S-Hop run even faster when the query is highly selective; i.e., smaller $k$ or larger $\tau$, which tend to be the more practical and meaningful query settings that people would use in real-life applications.
While S-Band is also a reasonable approach, its performance depends highly on the data characteristics (faring poorly in high dimensions and for certain distributions).
S-Band also requires additional offline indexing for finding durable $k$-skyband candidates. 
Overall, as demonstrated by experiments on both real and synthetic data, efficiency and robustness of Hop-based solutions make them more attractive solutions.
Even on very large and high-dimensional datasets, T-Hop/S-Hop only need less than a second to return durable top records for any given preference, which enables interactive data exploration.
Finally, T-Hop can be efficiently implemented inside a DBMS; for large datasets (tens of Gigabytes), it brings down the query time to just a couple of seconds, from more than 10 minutes required without our solution.
\section{Related Work}\label{sec:related-work}
The notion of ``durability'' on temporal data has been studied by previous works, but they consider different definitions of durability and/or different data models from ours.
In \cite{jiang2011prominent} and \cite{zhang2014discovering}, authors implicitly considered ``durability'' in the form of prominent streaks in sequence data, and devised efficient algorithms for discovering such streaks.
Given a sequence of values, a prominent streak is
a long consecutive subsequence consisting of only large (small) values.
Their algorithms can also be extended to find general top-k, multi-sequence and multi-dimensional prominent streaks.
Jiang and Pei~\cite{jiang2009online} studied Interval Skyline Queries on time series, which can be viewed as another type of ``durability'' when segments of time series dominate others.

Another line of durability-related work on temporal data is represented by~\cite{lee2009consistent, wang2014durable, gao2018durabletopk} and~\cite{mamoulis2010durable}.
Consider a time-series dataset with a set of objects, where the data values of each object are measured at regular time intervals; i.e., stock markets.
At each time $t$, objects are ranked according to their values at $t$.
The definition of ``durability'' therein is the fraction of time during a given time window when an object ranks $k$ or above.
This line of work mainly focused on how to efficiently aggregate rankings (rank $\leq k$ or not) over time.
\cite{mamoulis2010durable} applied durable top-$k$ searches in document archives, finding documents that are consistently among the most relevant to query keywords throughout a given time interval. 
In that setting, the challenge is how to merge multiple per-keyword relevance scores over time efficiently into a single rank. 

Durable queries also arise in dynamic or temporal graphs, typically represented as sequences of graph snapshots.
For example, in~\cite{semertzidis2016durable} and \cite{semertzidis2018top}, authors considered the problem of finding the (top-$k$) most durable matches of an input graph pattern query; that is, the
matches that exist for the longest period of time. 
The main focus is on the representations and indexes of the sequence of graph snapshots, and how to adapt classic graph algorithms in this setting.


Besides durability, Mouratidis et al.~\cite{mouratidis2006continuous} studied how to continuously monitor top-$k$ results over the most recent data in a streaming setting.
Our baseline solution used in Section~\ref{sec:expr} shares the same spirit as algorithms in~\cite{mouratidis2006continuous} for incrementally maintaining top-$k$ results over consecutive sliding windows.

\section{Conclusion}\label{sec:conclusion}
In this paper, we have initiated a comprehensive study into the problem of finding durable top records in large instant-stamped temporal datasets by running durable top-$k$ queries.
We proposed two types of novel algorithms for efficiently solving this problem, and provided in-depth theoretical analysis on the complexity of the problem itself and of our algorithms.
As demonstrated by experiments on real and synthetic data, our best solutions, Time-Hop and Score-Hop, find interesting durable top records in under a second on large and high-dimensional datasets, and can be up to 2 orders of magnitude faster than existing baselines.

\bibliographystyle{IEEEtran}
\bibliography{bibliography}

\newpage
\appendix
\section{Implementation details of top-$k$ details and query algorithm}

\subsection{Implementation Details}\label{sec:pre:implementation}

For simplicity and usability, we adopt a more straightforward tree-based implementation that better serves our purpose for answering a preference top-$k$ query in a time window. 

\begin{algorithm}[h]
    \KwIn{Dataset $P$}
    \KwOut{A Tree Index $\mathcal{T}$ for Preference Top-$k$ Query}
    \SetKwProg{Fn}{Def}{:}{}
     \Fn {BuildTree$(t_1, t_2)$} {
     \If {$t_1 > t_2$} {
     \KwRet null\;
     }
     \ElseIf {$t_1 == t_2$} {
     create a leaf node $n$\;
     $n.\text{skyline} \leftarrow P[t_1]$\;
     $n.\text{interval} \leftarrow [t_1,t_2]$\;
     \KwRet $n$\;
     }
     \Else{
     create a node $n$\;
     $t_m \leftarrow t_1 + (t_1+t_2)/2$\;
     $n.\text{left\_child} \leftarrow$ BuildTree$(t_1, t_m)$\;
     $n.\text{right\_child} \leftarrow$ BuildTree$(t_m+1, t_2)$\;
     $n.\text{interval} \leftarrow [t_1,t_2]$\;
     $n.\text{skyline} \leftarrow \kskyband{1}{\kskyband{1}{P([t_1, t_m])} \cup \kskyband{1}{P([t_m+1, t_2])}}$\; 
     \KwRet $n$\;
     }
     }
     \caption{Tree Index Construction\label{algo:topk-index}}
\end{algorithm}

\begin{algorithm}[h]
    \KwIn{$P$, $\mathcal{T}$,$\wtv$, $k$, and $\Qinterval$}
    \KwOut{$\topk{\leq k}{\wtv}{\Qinterval}$}
    \SetKwProg{Fn}{Def}{:}{}
     \Fn {PreferenceTopK$(\Qinterval, \wtv, k)$} {
     candidates $\leftarrow \emptyset$\;
     $\mathcal{Q}$ (priority queue in descending order of key) $\leftarrow \emptyset$\;
     $N \leftarrow$ a set of canonical nodes from $\mathcal{T}$ that covers $\Qinterval$\;
     \For {$n_i \in N$} {
        Compute interval max score $v_i$ using $n_i$.skyline\;
        $\mathcal{Q}$.push$(v_i, n_i)$\;
     }
     \While {$\card{\text{candidate}} < k$ and $!\mathcal{Q}$.empty()} {
        $v, n \leftarrow \mathcal{Q}$.top(), $\mathcal{Q}$.pop()\;
        \If {$\card{n\text{.interval}} > $ LENGTH\_THRESHOLD} {
            $n_l \leftarrow n$.left\_child, $n_r \leftarrow n.$right\_child\;
            Compute $v_l, v_r$ from $n_l, n_r$ using skylines\;
            $\mathcal{Q}$.push$(v_l, n_l)$, $\mathcal{Q}$.push$(v_r, n_r)$\;
        }
        \Else {
            candidate.push($n$)\;
        }
     }
     Compute $\topk{\leq k}{\wtv}{\Qinterval}$ using candidates\;
     \KwRet $\topk{\leq k}{\wtv}{\Qinterval}$\;
     }
 \caption{Preference Top-$k$ Query $Q(\wtv, k , \Qinterval)$ \label{algo:topk-query}}
\end{algorithm}

Consider a query time $W$ decomposed into $n$ non-empty disjoint time intervals $W = \bigcup_{i=1}^{n}\Qinterval_i$. 
Assume for each interval $\Qinterval_i$ we know the highest score (with respect to $\wtv$) among $P(\Qinterval_i)$, referred to as \emph{interval max score}.
It is sufficient to use at most $k$ out of $n$ intervals \footnote{Using all records in $P$ that arrive during these $k$ time intervals to compute the top-$k$ results.} to answer a preference top-$k$ query $Q(\wtv, k, W)$ if the chosen intervals have the $k$ highest interval max scores.
Based on this idea, our implementation takes advantages of two important properties of \emph{skyline}\cite{borzsony2001skyline} to improve the efficiency of index construction and query procedure.

As shown in Algorithm~\ref{algo:topk-index}, the tree index is built upon the dimension of arriving time of all points in $P$ in a bottom-up manner.
Each leaf node corresponds to a single timestamp (Line 6) and each internal node represents a time window (Line 14).
Each tree node also contains a skyline of points arriving during its window.
Skylines in all internal nodes can be efficiently computed from bottom to up (Line 15).

Algorithm~\ref{algo:topk-query} specifies the query procedure using the tree index.
Starting from the canonical intervals (nodes) of query window $\Qinterval$ (Line 4), we recursively split long intervals \footnote{The pre-determined value of LENGTH\_THRESHOLD controls the granularity of the chosen $k$ intervals for preference top-$k$ computations. By default, we set LENGTH\_THRESHOLD=128.} into smaller ones (Line 10-13), and use a priority queue to remember at most $k$ intervals that have the highest interval max scores (Line 15).
Finally, a preference top-$k$ result is computed using at most $k$ such intervals and all corresponding records in $P$ (no more than $k * $LENGTH\_THRESHOLD in total).
We can efficiently compute the interval max score for any interval $\Qinterval$ (Line 6 and 12). 

\subsection{Missing Proofs of Section~\ref{sec:time-solution}}\label{appendix:proof1}
\remove{
\begin{figure}
    \centering
    \tikzset{every picture/.style={line width=0.75pt}} 

\begin{tikzpicture}[x=0.75pt,y=0.75pt,yscale=-1,xscale=1]

\draw    (127,154) -- (354.61,154) ;
\draw [shift={(356.61,154)}, rotate = 180] [color={rgb, 255:red, 0; green, 0; blue, 0 }  ][line width=0.75]    (10.93,-3.29) .. controls (6.95,-1.4) and (3.31,-0.3) .. (0,0) .. controls (3.31,0.3) and (6.95,1.4) .. (10.93,3.29)   ;

\draw [line width=3.75]    (168,154) -- (328.3,154) ;

\draw  [color={rgb, 255:red, 0; green, 0; blue, 0 }  ,draw opacity=1 ][fill={rgb, 255:red, 0; green, 0; blue, 0 }  ,fill opacity=1 ] (307.11,153.77) .. controls (307.12,150.23) and (310,147.37) .. (313.54,147.38) .. controls (317.09,147.39) and (319.95,150.27) .. (319.94,153.81) .. controls (319.93,157.35) and (317.05,160.21) .. (313.51,160.2) .. controls (309.97,160.19) and (307.1,157.31) .. (307.11,153.77) -- cycle ;
\draw  [fill={rgb, 255:red, 208; green, 2; blue, 27 }  ,fill opacity=1 ] (283.89,149) -- (294,149) -- (294,159.11) -- (283.89,159.11) -- cycle ;
\draw  [fill={rgb, 255:red, 208; green, 2; blue, 27 }  ,fill opacity=1 ] (235.89,149) -- (246,149) -- (246,159.11) -- (235.89,159.11) -- cycle ;
\draw  [fill={rgb, 255:red, 208; green, 2; blue, 27 }  ,fill opacity=1 ] (188.89,149) -- (199,149) -- (199,159.11) -- (188.89,159.11) -- cycle ;
\draw   (155.3,174.2) .. controls (155.3,178.87) and (157.63,181.2) .. (162.3,181.2) -- (227.34,181.2) .. controls (234.01,181.2) and (237.34,183.53) .. (237.34,188.2) .. controls (237.34,183.53) and (240.67,181.2) .. (247.34,181.2)(244.34,181.2) -- (307.3,181.2) .. controls (311.97,181.2) and (314.3,178.87) .. (314.3,174.2) ;
\draw  [fill={rgb, 255:red, 7; green, 117; blue, 246 }  ,fill opacity=1 ] (142.15,147) -- (148.3,157.2) -- (136,157.2) -- cycle ;
\draw   (288.3,138.2) .. controls (288.3,133.53) and (285.97,131.2) .. (281.3,131.2) -- (217.3,131.2) .. controls (210.63,131.2) and (207.3,128.87) .. (207.3,124.2) .. controls (207.3,128.87) and (203.97,131.2) .. (197.3,131.2)(200.3,131.2) -- (138.3,131.2) .. controls (133.63,131.2) and (131.3,133.53) .. (131.3,138.2) ;
\draw [line width=2.25]  [dash pattern={on 2.53pt off 3.02pt}]  (168,77.2) -- (168,157) ;

\draw [line width=2.25]  [dash pattern={on 2.53pt off 3.02pt}]  (328.3,78.2) -- (328.3,154) ;

\draw    (170,95) -- (326.3,95) ;
\draw [shift={(328.3,95)}, rotate = 180] [fill={rgb, 255:red, 0; green, 0; blue, 0 }  ][line width=0.75]  [draw opacity=0] (8.93,-4.29) -- (0,0) -- (8.93,4.29) -- cycle    ;
\draw [shift={(168,95)}, rotate = 0] [fill={rgb, 255:red, 0; green, 0; blue, 0 }  ][line width=0.75]  [draw opacity=0] (8.93,-4.29) -- (0,0) -- (8.93,4.29) -- cycle    ;

\draw (354,164) node  [align=left] {$\displaystyle t$};
\draw (247,192) node  [align=left] {$\displaystyle \tau $};
\draw (216,120) node  [align=left] {$\displaystyle \tau $};
\draw (244,86) node  [align=left] {$\displaystyle \tau $};
\draw (316,166) node  [align=left] {$\displaystyle p_{1}$};
\draw (291,165) node  [align=left] {$\displaystyle p_{2}$};
\draw (242,165) node  [align=left] {$\displaystyle p_{3}$};
\draw (195,165) node  [align=left] {$\displaystyle p_{4}$};
\draw (144,164) node  [align=left] {$\displaystyle p_{5}$};

\end{tikzpicture}
    \caption{On a fixed $\tau$-length window, a false check could reduce the number of top-$k$ records that come from this window by 1. }
    \label{fig:asw-proof}
\end{figure}
}

\begin{proof}[Proof of Lemma~\ref{lemma:asw-bound}]
Let $\Qinterval=[a,b]$ and $\rho=[b-\tau, b]$.
Let $S_\rho=S\cap \rho$, i.e., the set of durable records with timestamp in $\rho$. We show that the number of false checks in $\rho$ is $O(\cardin{S_\rho}+k)$. Without loss of generality, assume that for any pair of records $p_i, p_j$ with $i<j$, $p_i.t<p_j.t$.

We consider two types of false checks in $\rho$.
If the algorithm finds a false check immediately after a durable record then this is a type-1 false check. Otherwise it is a type-2 false check.
From the definition, the number of type-1 false checks is bounded by $O(\cardin{S_\rho})$. Next we show that the number of type-2 false checks in $\rho$ is bounded by $O(k)$. If the number of records in $\rho$ is less than $k$ then the result follows, so we assume that $\cardin{P[b-\tau, b]}> k$.

Recall that if the algorithm visits a record $p$ it computes the top-$k$ elements in $[p.t-\tau, p.t]$. Let $U_p$ be the list of the top-$k$ items in $[p.t-\tau, p.t]$.
Let $Z_p=U_p\cap \rho$, be the list of these top-$k$ elements that lie in $\rho$. Generally we refer to $Z_p$ as a $Z$ list. At the beginning of the algorithm assume that we find the top-$k$ elements in a window of length $\tau$ from the rightmost item in $\rho$, so we have a list $Z$ with $\cardin{Z}\leq k$. We show the following two observations. i) Each time that the algorithm finds a type-2 false check the new $Z$ list of top-$k$ records in $\rho$ has cardinality at least one less than the previous list. ii) The cardinalities of the $Z$ lists as we run the algorithm in $\rho$ are never increasing.
If we show (i), (ii) we could argue that after the algorithm finds $k$ type-2 false checks in $\rho$, the $Z$ list will be empty and the algorithm will visits a record out of $\rho$.

Without loss of generality, assume that the rightmost record in $\rho$ was a type-2 false check. Let $Z_r$ be the current list as defined above. The algorithm visits the record with the largest timestamp in $Z_r$, say $p$, which is a type-2 false check. Let $Z_p=U_p\cap \rho$ be the new list. We compare the new list $Z_p$ with the old list $Z_r$. Notice that every record $q\notin Z_r$ with time $q.t\in[b-\tau,p.t]$ has $\score(q)<\score(p)$ (1), otherwise $Z_r$ would not be in the correct top-$k$ list. Furthermore, $p$ is a false check because there are at least $k$ records in $[p.t-\tau, p.t)$ with score larger than the score of $p$, (2). From (1), (2) it follows that $Z_p\subset Z_r$. Hence, the cardinality of the new $Z$ list is less than the cardinality of the previous $Z$ list.
In addition, notice that there are at least $k-\cardin{Z_p}$ records in $[p.t-\tau, b-\tau]$ with scores greater than the score of $p$, and generally greater than the score of any record in $P[b-\tau, p.t]\setminus Z_p$, (3).

In order to complete the proof we need to show what is the new $Z$ list when the algorithm visits a series of durable records.
Assume that $Z_p$ is the current list (or the initial one) and the algorithm visits $Z_p$'s record with the larger timestamp. Assume that the algorithm finds a series of durable records, where $j$ of them belong in $Z_p$.
Notice that $j\geq 1$. Let $q$ be the type-1 false check that the algorithm visits (after the series of durable records) and let $Z_q$ be the new list. We need to show that $\cardin{Z_q}\leq \cardin{Z_p}$. We assume that $q\notin Z_p$ (if $q\in Z_p$ then notice that $Z_q\subset Z_p$ so the result follows). Recall from (3) that there are at least $k-\cardin{Z_p}$ records with timestamp $[p.t-\tau>q.t-\tau,b-\tau]$ and with score greater than the score of $q$. We call these records $A$. Moreover, there are $\cardin{Z_p}-j$ records in $Z_p$ with timestamp in $[b-\tau, q.t)$ and with score greater than the score of $q$. We call these records $B$.
We have $\cardin{Z_q}\leq \cardin{B}+(k-\cardin{A}-\cardin{B})=k-\cardin{A}\leq \cardin{Z_p}$.
Hence, we conclude that there are $O(k)$ type-2 false checks and the total number of false checks in $\rho$ is $O(\cardin{S_\rho}+k)$.

There are $\ceiling{\frac{\card{\Qinterval}}{\tau}}$ intervals of length $\tau$ in $\Qinterval$ so the total number of false checks is $O(\cardin{S}+k\ceiling{\frac{\card{\Qinterval}}{\tau}})$.
\end{proof}


\newcommand{\dens}{dens}
\subsection{Missing Proofs of Section~\ref{sec:weight-solution}}
\label{appendix:proof2}

We first introduce some useful notation.
Let $\dens(t)$ be the density of a timestamp $t$, i.e., the number of blocking intervals that contain $t$. Notice that $\dens(t)$ is changing as we execute the algorithm.
If a record $p_i$ is blocked by at least $k$ records, i.e., $\dens(p_i.t)\geq k$, at line 7 of Algorithm~\ref{algo:psi} then we call it an \emph{auxiliary record}. Overall, we have that a record can be a durable record, a false check (we run a top-$k$ query but the record does not belong in the solution), or an auxiliary record.

We first start with a lemma that will be useful later.
\begin{lemma}\label{apndx:lem1}
Let $M_i$ be a set that is empty after the algorithm considering a (auxiliary) record from $M_i$ with density at least $k$, and let $[l_i, r_i]$ be its corresponding sub-interval. Then one of the two cases hold: The density of each timestamp in $[l_i, r_i]$ is at least $k$ or the algorithm has visited all records in $P([l_i, r_i])$.
\end{lemma}
\begin{proof}
If $\card{P([l_i, r_i])}\leq k$ then the algorithm visits all records in $P([l_i, r_i])$, since we always consider the top-$k$ records in $[l_i, r_i]$. If $\card{P([l_i, r_i])}> k$ then we show that when $M_i$ is empty every timestamp in $[l_i, r_i]$ has density at least $k$.

We prove the following argument by induction: When the algorithm visits a new auxiliary record $p_j$ in a set $M_j$ then any timestamp in $[l_j, t_j]$ has density at least $k$. Let $p_1$ be the first auxiliary record that the algorithm finds and let $M_{i_1}$ be the set that it belongs to. Since $p_1$ is an auxiliary record we have that $\dens(p_1.t)\geq k$ at the moment we visit $p_1$. Furthermore, notice that the algorithm did not consider any other record in $[l_{i_1}, p_1.t]$ in a previous iteration so we can argue that the density of every record in $[l_{i_1}, p_1.t]$ is at least $k$. In addition, notice that it is not possible to find any durable record or any false check in $[l_{i_1}, p_1.t]$ in the future. As a result, if we visit $p_1$ again in the future it will be an auxiliary record in a set with left endpoint the same $l_{i_1}$ timestamp. Let $p_{h-1}$ be an auxiliary record that the algorithm visits in set $M_{i_{h-1}}$ and let assume that any record in $[l_{i_{h-1}}, p_{h-1}.t]$ has density at least $k$. Let $p_h$ be the next auxiliary record that the algorithm visits and let assume that it belongs in a set $M_{i_h}$. First assume that the algorithm has visited $p_h$ in a previous iteration. Let $M_f$ be the set that contained $p_h$ when the algorithm first visited $p_h$. At the moment when the algorithm first visited $p_h$, we had that $\dens(p_h.t)\geq k$ and from the induction hypothesis we have that every timestamp in $[l_f,p_h.t]$ had density at least $k$. Hence, there was no other durable record or false check in $[l_f, p_h.t]$ in the future. That means that $l_f=l_{i_h}$ and so it holds that every record in $[l_{i_h},p_h.t]$ has density at least $k$.
Next, assume that this is the first time that we visit the auxiliary record $p_h$.
If this is the first auxiliary record in $M_{i_h}$ we have that the density of every record in $[l_{i_h}, p_h.t]$ has density at least $k$ because $\dens(p_h.t)\geq k$ and there is no subinterval that starts in $[l_{i_h}, p_h.t]$. Then, we study the case where $p_h$ is not the first auxiliary record that the algorithm finds in set $M_{i_h}$. Let $p_u$ be the auxiliary record in $M_{i_h}$ with the largest timestamp just before the algorithm found $p_h$.
From induction hypothesis we know that the density of every record in $[l_{i_h}, p_u.t]$ is at least $k$. If $p_h.t\leq p_u.t$ then $[l_{i_h}, p_h.t]\subseteq [l_{i_h}, p_u.t]$ so any record in $[l_{i_h}, p_h.t]$ has density at least $k$. The last case to consider is when $p_h.t>p_u.t$. Since $\dens(p_h.t)\geq k$, and since there is no sub-inerval that starts in $(l_u,p_h.t)$ we have that every record in $[l_u, p_h.t]$ has density at least $k$. We conclude that the density of every timestamp in $[l_{i_h}, p_h.t]$ is at least $k$.

Now we are ready to prove our lemma. If $\card{P([l_i, r_i])}> k$ and $M_i$ is empty it means that the algorithm has already considered $k$ auxiliary records in $[l_i, r_i]$. Let $p_u$ be the auxiliary record in $M_i$ with the largest timestamp. From the induction we have that the density of every record in $[l_i, p_u.t]$ is at least $k$. Furthermore, the algorithm has visited $k$ auxiliary records and hence it has added at least $k$ blocking intervals with left endpoint in $[l_i, p_u.t]$. All the intervals we add have length $\tau$ and $r_i-l_i\leq \tau$ so all timestamps in the interval $[p_u.t, r_i]$ have density at least $k$. We conclude that the density of each record in $[l_i, r_i]$ is at least $k$.
\end{proof}

\begin{proof}[Proof of Lemma~\ref{lemma:psi:correctness}]
Let $S^*$ be the durable records in $\Qinterval$.
We show that $S\subseteq S^*$ and $S^*\subseteq S$ showing that $S=S^*$.
The algorithm always checks by running a top-$k$ query if a record should be in the solution (line 8 of Algorithm~\ref{algo:psi}) so $S\subseteq S^*$.

Next we show the other direction.
The algorithm visits the records in descending (on score) order so it is not possible that a record $p\in S^*$ is blocked by at least $k$ records before the algorithm visits $p$.
Before we argue that $S^*\subseteq S$ we also need to make sure that the algorithm does not miss any durable record in a sub-interval $[l_j, r_j]$  that corresponds to an empty set $M_j$.
In Lemma~\ref{apndx:lem1} we showed that all timestamps in $[l_j, r_j]$ have density at least $k$ so there is no additional durable record in this sub-interval. Hence $S^*\subseteq S$, and overall we conclude that $S=S^*$.
\end{proof}

Let $p_i$ be a false check that the algorithm just found, and let $P_i'$ be the top-$k$ records in $[p_i.t-\tau, p_i.t)$, as we had in the algorithm. Let $p_i'$ be the record in $P_i'$ with the largest timestamp. We say that $p_i$ is \emph{assigned} to $p_i'$. If $p_i.t'<a$, where $a$ is the timestamp such that $\Qinterval=[a, b]$, then $p_i$ is assigned to $a$.
The next lemma follows from the definition.
\begin{lemma}\label{apndx:lem2}
Assume that the algorithm just found the false check $p_i$. After adding all the blocking intervals from $P_i'$ we have that the density of every timestamp in $[p_i'.t, p_i.t]$ is at least $k$.
\end{lemma}
We show the next lemma which is useful to bound the number of false checks.
\begin{lemma}\label{apndx:lem3}
Let $p_i$ be a false check and $p_i'$ be the record that it is assigned to. Before adding the $k$ blocking intervals from all records in $P_i'$ (as defined above) we have that either $\dens(p_i'.t)\geq k$, or $p_i'\in S$ and $\dens(p_i'.t)<k$, or $p_i'=a$.
\end{lemma}
\begin{proof}
If $p_i'.t<a$ then from the definition $p_i'$ is $a$. (Notice that if we find more than one false checks that are assigned to $a$ then $\dens(a)>k$, so this case can be considered the same as $\dens(p_i'.t)>k$.)

Next, we assume that $p_i'.t\geq a$.
We prove the lemma by contradiction.
Let $p_i'$ be a record that does not belong in $S$ and $\dens(p_i'.t)<k$.
Notice that $\score(p_i')>\score(p_i)$.
Since $p_i'$ is not in $S$ it can be either: a false check, an auxiliary record, or a record that the algorithm has not visited before.
If $p_i'$ is a false check then from Lemma~\ref{apndx:lem2} we have that $\dens(p_i'.t)\geq k$ at the moment that we found $p_i'$ for first time, which is a contradiction.
If $p_i'$ is an auxiliary record then from Lemma~\ref{apndx:lem1} we have that $\dens(p_i'.t)\geq k$, which is a contradiction.
If $p_i'$ is a record that the algorithm has not considered before then there are two cases:
a) $p_i'$ belongs in an interval $[l_j, r_j]$ of a set $M_j$ that we have removed from $M$ because we have already visited its top-$k$ records.
From Lemma~\ref{apndx:lem1} we know that $\dens(p_i'.t)\geq k$, which is a contradiction.
b) $p_i'$ belongs in an interval $[l_j, r_j]$ of a set $M_j$ that there still exists in $H$. Since $\score(p_i')>\score(p_i)$ it means that $p_i$ is not the record with the highest score among the sub-intervals that are not removed from $M$, which is a contradiction.

In any case we proved that either $p_i'$ has density at least $k$, or $p_i'$ has density less than $k$ and $p_i'\in S$, or $p_i'=a$.
\end{proof}

\begin{proof}[Proof of Lemma~\ref{lemma:ps-bound}]
If a false check $p_i$ is assigned to a durable record with density less than $k$ then we call it type-1 false check.
Otherwise, it is a type-2 false check.

Let $p_i$ be a type-1 false check so we have that $p_i'\in S$ and $\dens(p_i'.t)<k$. After adding all the $k$ segments from $P_i'$ we have that $\dens(p_i'.t)\geq k$. The next time that $p_i'$ will be assigned by another false check the density of $p_i'$ will be at least $k$ so it will be a type-2 false check. Hence, it is straightforward to bound the number of type-1 false checks, which is at most $O(\card{S})$.

Next we focus on type-2 false checks.
Let $[l, r]$ be one of the initial disjoint $\tau$-length windows from line 2 of Algorithm~\ref{algo:psi}.
We show that after finding $k$ type-2 false checks in $[l, r]$ the density of all timestamps in $[l, r]$ is at least $k$. If that is the case then the algorithm will not find any other false check in $[l, r]$.

Let $t$ be any timestamp in $[l, r]$. We show that $\dens(t)\geq k$ after finding $k$ type-2 false checks in $[l, r]$.
If one of the false checks in $[l, r]$ lies on $t$ then we already have that $\dens(t)\geq k$.
Let assume that the algorithm finds $k_1$ type-2 false checks in $[l,t)$ and $k_2$ type-2 false checks in $(t, r]$, where $k_1+k_2=k$.
If $k_1\geq k$ then $\dens(t)\geq k$, so the interesting case is when $k_1<k$ and $k_2\geq 1$.
Let $\chi$ be the total number of blocking intervals that the algorithm has added having their right-endpoint in $[t, r]$ after finding all the $k$ type-2 false checks in $[l, r]$, and let $X$ be the set of those intervals.
We have that $\dens(t)\geq k_1+\chi$. We show that $k_1+\chi\geq k$ or equivalently $k_2\leq \chi$.

Let $p_i$ be a type-2 false check that the algorithm just found in $(t, r]$. Let $p_i'$ be the record that $p_i$ is assigned to, as we defined above. If $p_i'.t\leq t$ then we immediately have that $\dens(t)\geq k$ after adding the at most $k$ new segments from the set $P_i'$ (Lemma~\ref{apndx:lem2}), so this case is not interesting.
(Notice that if $p_i'.t<a$, before $p_i'$ is set to be $a$, then this is always the case since $t\geq a$).

Now, we assume that for each $p_i$ which is a type-2 false check in $(t,r]$, it holds that $p_i'.t\in (t, p_i.t)$.
The main idea to prove that $k_2\leq \chi$ is the following: Each time that the algorithm finds a type-2 false check in $(t,r]$ we find an unmarked interval in $X$ and we mark it.
In particular, we show that there always be such an unmarked segment in $X$ with its right endpoint in $[p_i'.t, p_i.t)$.
Since $p_i$ is a type-2 false check we have that $\dens(p_i'.t)\geq k$ and $\dens(p_i.t)<k$,
at the moment that the algorithm visits $p_i$ (before adding the at most $k$ segments from $P_i'$).
Let $Z_1$ be the current blocking intervals with right endpoint in $[p_i'.t,p_i.t)$ and $z_1=\card{Z_1}$. Let $Z_2$ be the current blocking intervals with left endpoint in $(p_i'.t, p_i.t]$, and $z_2=\card{Z_2}$. Let $B$ be the current blocking intervals with left endpoint in $[p_i.t-\tau, p_i.t]$.
We have that $\dens(p_i.t)<k\Leftrightarrow \card{B}<k$, (1). We also have $\dens(p_i'.t)\geq k\Leftrightarrow \card{B}-z_2+z_1\geq k$, (2).
From (1), (2), we have that $z_1>z_2\Leftrightarrow z_1\geq z_2+1$.
By definition, notice that the false checks with time instance $\leq p_i'.t$ cannot mark a segment in $Z_1$.
Furthermore, a previous false check with timestamp at the right of $p_i.t$ cannot mark a segment in $Z_1$: Let $p_j$ be a false check that the algorithm found in a previous iteration in $(p_i.t, r]$ and let $p_j'$ be the record that it is assigned to. If $p_j'.t>p_i.t$ then the marking process does not mark any segment in $Z_1$. Otherwise, if $p_j'.t\leq p_i.t$ then the density of all records in $[p_j'.t,p_i.t]\cup [p_i.t, p_j'.t]$ would be at least $k$ after the algorithm adds the segments from $P_j'$, which is a contradiction because $\dens(p_i.t)<k$ when we visit $p_i$.
Hence only false checks in $(p_i'.t, p_i.t]$ can mark segments in $Z_1$.
Recall that $Z_2$ are the current segments with left endpoints in $(p_i'.t, p_i.t]$.
Even if all segments in $Z_2$ were created by type-2 false checks and even if all of them mark segments from $Z_1$, we showed that $z_1\geq z_2+1$, so we can always find a new unmarked segment in $Z_1$. Notice that any segment in $Z_1$ has its right endpoint in $[p_i'.t,p_i.t)$ and since all the segments have length $\tau$, they contain $t$ and hence they belong in $X$. Each time that we find a type-2 false check in $(t, r]$ we mark a new segment in $X$, so $k_2\leq \chi$ and we conclude that $\dens(t)\geq k$.

Recall that $t$ can be any record in $[l, r]$, so we showed that after finding $k$ type-2 false checks in $[l, r]$ the density of every timestamp in $[l,r]$ is at least $k$. As a result, the algorithm will not find any other false check in $[l, r]$.
There are at most $\big\lceil\frac{\card{\Qinterval}}{\tau}\big\rceil$ disjoint $\tau$-length windows in $\Qinterval$ so the number of type-2 false checks is bounded by $O(k\big\lceil\frac{\card{\Qinterval}}{\tau}\big\rceil)$
The overall number of false checks along with the durable records is $O(\card{S}+k\big\lceil\frac{\card{\Qinterval}}{\tau}\big\rceil)$.
\end{proof}

\subsection{Missing Proofs of Section~\ref{sec:answer-size}}\label{appendix:proof3}

\begin{proof}[Proof of Lemma~\ref{lemma:expected-Skyband-size}]
We show the result extending the main ideas from \cite{bentley1977average}.
Let $P(\Qinterval)=\{\pnt_{j+1},\ldots, \pnt_{j+L}\}$. For $\pnt_{i}\in P(\Qinterval)$, let $X_i$ be a random variable which is $1$ if $\pnt_i\in \mathcal{C}$, and $0$ otherwise.
From linearity of expectation we have that $\Expec{\cardin{C}}=\Expec{\sum_{i=j+1}^{j+\card{\Qinterval}}X_i}=\sum_{i=j+1}^{j+\card{\Qinterval}}\Expec{X_i}=\sum_{i=j+1}^{j+\card{\Qinterval}}\Prob{X_i=1}$.
We focus on computing $\Prob{X_i=1}$.
Let $P_i=P([\pnt_{i-\tau}.t, \pnt_i.t])=\{\pnt_{i-\tau}, \ldots, \pnt_{i-1}, \pnt_i\}$.
By independence we have that the probability of each point in
$P_i$ to be in the $k$-skyband of $P_i$ is the same,
so we can compute $\Prob{X_i=1}$ by first finding the expected size of the $k$-skyband in $P_i$ and then divide it by the number of points, $\tau+1$.

Let $B_i$ be the $k$-skyband of the $\tau+1$ points $P_i$.
Let $V_j\subset N$ for $1\leq j\leq d$, with $\cardin{V_j}=\tau+1$ such that $V_j$ contains the values that are assigned to the $j$-th coordinate of the points in $P_i$.
We compute $\Expec{\cardin{B_i}\mid V_1, \ldots, V_d}$.
Let $A(\tau+1,d)$ be the expected size of the $k$-skyband of a set $\bar{P}$ with $\tau+1$ points in $\Re^d$ in the $d$-dimensional random permutation model. Notice that $A(\tau+1,d)=\Expec{\cardin{B_i}\mid V_1, \ldots, V_d}$.
We compute $A(\tau+1,d)$ as follows. From linearity of expectation we can compute the  probability that a point in $\bar{P}$ belongs in the $k$-skyband and take the sum of them,
$A(\tau+1,d)=\sum_{\bar{p}\in \bar{P}}\Prob{\bar{p}\in k\text{-skyband of } \bar{P}}$.
Assume that a point $\bar{p}\in \bar{P}$ has the $g$-th largest first coordinate among the points in $\bar{P}$. Notice that this can happen with probability $\frac{1}{\tau+1}$.
Since the first coordinate of the $g$-th point ($\bar{p}$) is greater than the first coordinates of $g-1$ points it cannot be dominated by any of those. Therefore, the $g$-th point belongs in the $k$-skyband if and only if its remaining $d-1$ coordinates belong in the $k$-skyband among the points in $\bar{P}$ with the $g$-th through the $(\tau+1)$-th largest first coordinate.
The probability that the $g$-th point is in the $k$-skyband is, by independence, the expected number of the $k$-skyband in the remaining points and coordinates, which is $A(\tau+1-g+1,d-1)$, divided by the total number of the remaining points in the set which are $\tau+1-g+1$. Notice that $A(k',y)=k'$ for $k'\leq k$ and any $y$.
Hence, we have $A(\tau+1,d)=\sum_{j=1}^{\tau+1}\sum_{g=1}^{\tau+1}\frac{1}{\tau+1}\frac{A(\tau+1-g+1, d-1)}{\tau+1-g+1}=\frac{1}{\tau+1}\sum_{j=1}^{\tau+1}\sum_{J=1}^{\tau+1}\frac{A(J, d-1)}{J}=\sum_{J=1}^{\tau+1}\frac{A(J, d-1)}{J}$.
Notice that $A(x,y)$ is monotonically increasing in $x$, so if $x_1\leq x_2$, then $A(x_1,y)\leq A(x_2,y)$ for any $y$.
Furthermore, we note that $A(\tau+1,1)=k$ since in one dimension the top-$k$ points belong in the $k$-skyband.
We have,
$A(\tau+1,d)=\sum_{J=1}^{\tau+1}\frac{A(J, d-1)}{J}\leq A(\tau+1,d-1)\sum_{J=1}^{\tau+1}\frac{1}{J}\leq A(\tau+1,d-1)O(\log \tau)$. Iterating this recurrence on $d$ until $A(\tau+1,1)=k$ gives the upper bound $A(\tau+1,d)=O(k\log^{d-1}\tau)$.

We conclude that $\Expec{\cardin{B_i}\mid V_1,\ldots, V_d}=O(k\log^{d-1}\tau)$.
Notice that $\Prob{V_1,\ldots, V_d}=\frac{1}{{n\choose \tau+1}^d}$ and all possible sets of $V_1,\ldots, V_d$ are ${n\choose \tau+1}^d$ so we have that $\Expec{\cardin{B_i}}=O(k\log^{d-1}\tau)$, and $\Prob{X_i=1}\sim \frac{O(k\log^{d-1}\tau)}{\tau+1}$.
Overall we conclude that $\Expec{\cardin{\mathcal{C}}}=\sum_{i=j+1}^{j+\card{\Qinterval}}\Prob{X_i=1}=O(\frac{k\card{\Qinterval}}{\tau}\log^{d-1}\tau)$.
\end{proof}

\end{document}